%% file: main.tex
\documentclass{article}
\usepackage{graphicx} 
\usepackage{amsthm}
\usepackage{amsfonts}
\usepackage{amsmath}
\usepackage{amssymb}
\usepackage[margin=1in]{geometry}
\usepackage{todonotes}
\usepackage{cite}

\newtheorem{theorem}[]{Theorem}
\newtheorem{lemma}[]{Lemma}
\newtheorem{prop}[]{Proposition}
\newtheorem{definition}[]{Definition}
\newtheorem{corollary}{Corollary}
\newtheorem{oq}{Open Question}

\newcommand{\dist}{\text{dist}}
\newcommand{\Oish}{\widetilde{O}}
\newcommand{\Omegaish}{\widetilde{\Omega}}
\newcommand{\hopdist}{\texttt{hopdist}}
\newcommand{\rr}{\mathbb{R}}
\newcommand{\zz}{\mathbb{Z}}

\newcommand{\eps}{\varepsilon}

\usepackage{todonotes}
\usepackage{xcolor}

\title{Folklore Sampling is Optimal for Exact Hopsets: Confirming the $\sqrt{n}$ Barrier\footnote{This work was supported by NSF:AF 2153680.}}
\author{
\begin{tabular}{c c}
   Greg Bodwin  & Gary Hoppenworth  \\
   University of Michigan & University of Michigan
   \\
     \texttt{bodwin@umich.edu}  &  \texttt{garytho@umich.edu} \\
\end{tabular}
}
\date{}

\begin{document}

\maketitle





\input{intro}

\input{overview}

\input{hopset}

\input{shortcut}

\bibliographystyle{plain}
\bibliography{refs}

\appendix

\input{appendix}

\end{document}

%% file: intro.tex
\begin{abstract}
For a graph $G$, a \emph{$D$-diameter-reducing exact hopset} is a small set of additional edges $H$ that, when added to $G$, maintains its graph  metric but guarantees that all node pairs have a shortest path in $G \cup H$ using at most $D$ edges.
A \emph{shortcut set} is the analogous concept for reachability rather than distances.
These objects have been studied since the early '90s, due to applications in parallel, distributed, dynamic, and streaming graph algorithms.

For most of their history, the state-of-the-art construction for either object was a simple folklore algorithm, based on randomly sampling nodes to hit long paths in the graph.
However, recent breakthroughs of Kogan and Parter [SODA '22] and Bernstein and Wein [SODA '23] have finally improved over the folklore algorithm for shortcut sets and for $(1+\eps)$-approximate hopsets.
For either object, it is now known that one can use $O(n)$ hop-edges to reduce diameter to $\Oish(n^{1/3})$, improving over the folklore diameter bound of $\Oish(n^{1/2})$.
The only setting in which folklore sampling remains unimproved is for exact hopsets.
Can these improvements be continued?

We settle this question negatively by constructing graphs on which any exact hopset of $O(n)$ edges has diameter $\Omegaish(n^{1/2})$.
This improves on the previous lower bound of $\Omega(n^{1/3})$ by Kogan and Parter [FOCS '22].
Using similar ideas, we also polynomially improve the current lower bounds for shortcut sets, constructing graphs on which any shortcut set of $O(n)$ edges reduces diameter to $\Omegaish(n^{1/4})$.
This improves on the previous lower bound of $\Omega(n^{1/6})$ by Huang and Pettie [SIAM J.\ Disc.\ Math.\ '18].
We also extend our constructions to provide lower bounds against $O(p)$-size exact hopsets and shortcut sets for other values of $p$; in particular, we show that folklore sampling is near-optimal for exact hopsets in the entire range of parameters $p \in [1, n^2]$.
\end{abstract}

\thispagestyle{empty}
\clearpage
\setcounter{page}{1}

\section{Introduction}

In graph algorithms, many basic problems ask to compute information about the shortest path distances or reachability relation among node pairs in an input graph.
In parallel, distributed, dynamic, or streaming settings, algorithm complexity often scales with the \emph{diameter} of the graph, i.e., the smallest integer $d$ such that every connected node pair has a path of at most $d$ edges.
Therefore, a popular strategy to optimize these algorithms is to add a few edges to the input graph in preprocessing, with the goal to reduce diameter while leaving the relevant distance or reachability information unchanged.
In the context of reachability, this set of additional edges is called a \emph{shortcut set}.

\begin{definition} [Shortcut Sets]
For a directed graph $G$, a $D$-diameter reducing shortcut set is a set of additional edges $H$ such that:
\begin{itemize}
\item Every edge $(u, v) \in H$ is in the transitive closure of $G$; that is, there exists a $u \leadsto v$ path in $G$.

\item For every pair of nodes $(s, t)$ in the transitive closure of $G$, there exists an $s \leadsto t$ path in $G \cup H$ using at most $D$ edges.
\end{itemize}
\end{definition}

Shortcut sets were introduced by Thorup \cite{Thorup92}, after they were used implicitly in prior work.
Many algorithmic applications of shortcut sets and their relatives were discovered in the following years \cite{UY91, KS97, HKN14a, HKN14b, HKN15, FN18, JLS19, Fineman19, GW20, BGW20, CFR20, KS21, ASZ20}, but actual \emph{constructions} of shortcut sets were elusive.
For most of their history, essentially the only known construction was the following simple algorithm: randomly sample a set $S$ of $|S|=n^{1/2}$ vertices, and add a shortcut edge between each pair of sampled nodes that lie in the transitive closure of the input graph.
To argue correctness: for any nodes $s,t$ in the graph where the shortest path $\pi(s, t)$ has length $\gg \Omegaish(n^{1/2})$, with high probability we sample nodes $u, v$ in $S$ that respectively hit a prefix and suffix of $\pi(s, t)$ of length $\Oish(n^{1/2})$.
Using the added shortcut edge $(u, v)$, we obtain an $s \leadsto t$ path of length $\Oish(n^{1/2})$.
This analysis gives:
\begin{theorem} [Folklore, \cite{UY91}]
Every $n$-node graph has a $\Oish(n^{1/2})$-diameter-reducing shortcut set on $O(n)$ edges.
\end{theorem}

Remarkably, despite its simplicity, the diameter bound of $\Oish(n^{1/2})$ achieved by the folklore sampling algorithm remained nearly unimproved for 30 years (log factors were removed in \cite{BRR10}, improving the diameter bound to $O(n^{1/2})$).
This led researchers to wonder if the bound could be improved in the exponent at all.
This was finally answered affirmatively in a recent breakthrough of Kogan and Parter \cite{KP22a}:
\begin{theorem} [~\cite{KP22a} ]
The folklore algorithm is \textbf{polynomially suboptimal} for shortcut sets.
In particular, every $n$-node graph has a $\Oish(n^{1/3})$-diameter-reducing shortcut set on $O(n)$ edges.
\end{theorem}

Kogan and Parter proved this theorem via an elegant construction based on sampling  vertices \emph{and} sampling from a set of carefully-chosen paths from the input graph.
Following this, there are two clear avenues for further progress.
First, the new diameter bound of $\Oish(n^{1/3})$ is still not necessarily tight.
It was still conceivable to improve diameter as far as $O(n^{1/6})$, at which point we encounter a lower bound construction of Huang and Pettie \cite{HP18} (improving on a classic construction of Hesse \cite{Hesse03}).
Second, many algorithms aim to compute exact or approximate \emph{shortest paths} of an input graph, rather than \emph{any path} as in the case of shortcut sets/reachability.
These algorithms benefit from shortcut-set-like structures that more strongly reduce the number of edges along (near-)shortest paths in the input graph.
Such a structure is called a \emph{hopset}:

\begin{definition} [Hopsets]
For a graph $G$ and $\eps \ge 0$, a $D$-diameter reducing $(1+\eps)$ hopset is a set of additional edges $H$ such that:
\begin{itemize}
\item Every edge $(u, v) \in H$ has weight $w(u, v) := \dist_G(u, v)$.

\item For every pair of nodes $(s, t)$ in the transitive closure of $G$, there exists an $s \leadsto t$ path $\pi(s, t)$ in $G \cup H$ that uses at most $D$ edges, and which satisfies $w(\pi(s, t)) \le (1+\eps) \cdot \dist_G(s, t)$.
\end{itemize}
When $\eps=0$, the path $\pi(s, t)$ is required to be an \emph{exact} shortest path in $G \cup H$, so we call $H$ an exact hopset.
\end{definition}

A nice feature of the folklore sampling algorithm is that it extends immediately to hopsets with no real changes.
This yields:
\begin{theorem} [Folklore]
Every $n$-node graph has a $\Oish(n^{1/2})$-diameter-reducing (exact or $(1+\eps)$) hopset on $O(n)$ edges.
\end{theorem}

Thus, the hunt is back on for a hopset construction algorithm that beats folklore sampling.
Kogan and Parter partially achieved this goal: they extended their shortcut set construction to also show a new diameter bound of $\Oish(n^{2/5})$ for $(1+\eps)$ hopsets \cite{KP22a}.
Bernstein and Wein \cite{BW23} then developed a clever extension of the Kogan-Parter construction, further improving the bound for $(1+\eps)$ hopsets to match the one achieved for shortcut sets:
\begin{theorem}[\cite{KP22a, BW23} ] 
The folklore algorithm is \textbf{polynomially suboptimal} for $(1+\eps)$ hopsets.
In particular, for all fixed $\eps > 0$, every (possibly directed and weighted) $n$-node graph has a $\Oish(n^{1/3})$-diameter-reducing-shortcut set on $O(n)$ edges.
\end{theorem}

Still, both of these improvements required $\eps > 0$, and so neither extended to \emph{exact} hopsets, which still remained as the last holdout where the folklore algorithm had not been improved.
The only progress for exact hopsets came on the lower bounds side, where a separate work of Kogan and Parter \cite{KP22a} showed a diameter lower bound of $\Omega(n^{1/3})$ (see also \cite{BHT22}).
Was it possible to translate the recent progress on shortcut sets and $(1+\eps)$ hopsets to exact hopsets, and finally move past folklore?

\subsection{Our Results}

\begin{figure} [h]
\begin{center}
\begin{tikzpicture}
  \draw[thick,->] (0,0) -- (7,0) node[below] {};
  \draw[thick,->] (0,0) -- (0,7) node[left] {};

  \coordinate (A) at (0,6) ;
  \fill (A) circle (2pt) node [above right = -0.1] {\small \color{gray} folklore};
  \coordinate (B) at (3,4);
  \fill (B) circle (2pt) node [above] {\small \color{gray} \cite{BW23}};
  \coordinate (C) at (6,4);
  \fill (C) circle (2pt) node [above] {\small \color{gray} \cite{KP22a}};

  \coordinate (D) at (0,6);
  \fill (D) circle (2pt);
  \coordinate (E) at (3,3);
  \fill (E) circle (2pt);
  \coordinate (F) at (6,3);
  \fill (F) circle (2pt);

  \coordinate (G) at (0,4);
  \fill (G) circle (2pt) node [above right = -0.1] {\small \color{gray} \cite{KP22}};
  \coordinate (H) at (3,2);
  \fill (H) circle (2pt) node [above] {\small \color{gray} \cite{HP18}};
  \coordinate (I) at (6,2);
  \fill (I) circle (2pt) node [above] {\small \color{gray} \cite{HP18}};

  \draw[blue] (A) -- (B) -- (C) node [below] {upper bounds};
  \draw[red] (D) -- (E) -- (F) node [below] {new lower bounds};
  \draw[red!50!gray, dashed] (G) -- (H) -- (I) node[below] {old lower bounds};

  \node at (0, -0.5) {Exact Hopsets};
  \node at (3, -0.5) {$(1+\eps)$ Hopsets};
  \node at (6, -0.5) {Shortcut Sets};
  \node at (3, -1.5) {\bf $O(n)$-size Edge Set vs.\ Weighted Input Graph};

  \node at (-2, 3) {\bf Diameter};
  \node at (-0.5, 6) {$n^{1/2}$};
  \node at (-0.5, 4) {$n^{1/3}$};
  \node at (-0.5, 2) {$n^{1/6}$};
\end{tikzpicture}
\end{center}
\caption{\label{fig:nbounds} State-of-the-art bounds for $O(n)$-size hopsets and shortcut sets, before and after this paper
}
\end{figure}
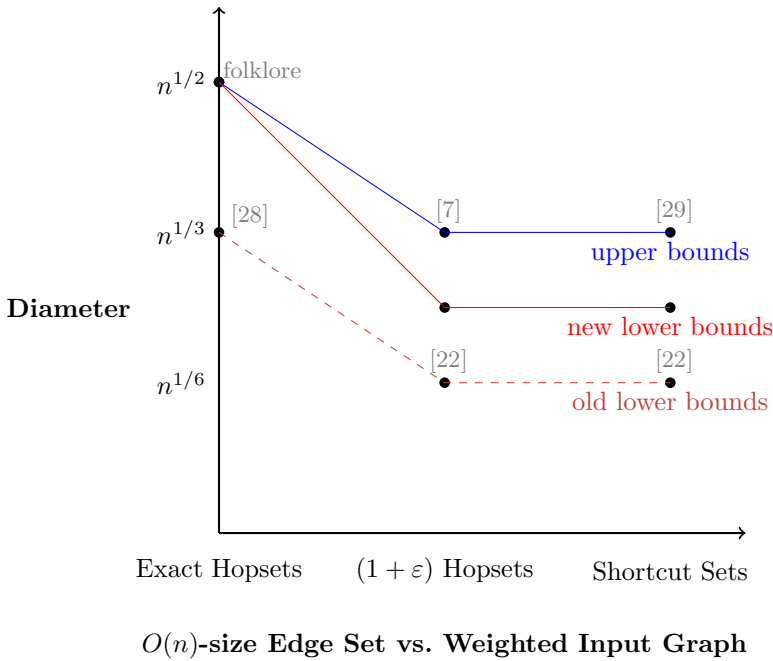

Our main results are polynomial improvements on the lower bounds for $O(n)$-size shortcut sets and hopsets, which we depict in Figure \ref{fig:nbounds}.
These lower bounds \emph{confirm} that the folklore algorithm for exact hopsets is essentially the right one, showing that its diameter bound is optimal up to $\log n$ factors:
\begin{theorem} [New] \label{thm:introeh}
The folklore algorithm is \textbf{near-optimal} for exact hopsets.
In particular, there are $n$-node graphs on which any exact hopset on $O(n)$ edges reduces diameter to $\Omegaish(n^{1/2})$.
\end{theorem}
This provides a strong separation between exact and $(1+\eps)$ hopsets.
Our lower bound holds for both directed and undirected input graphs, but it critically uses edge weights and thus does not extend to unweighted graphs as well.
Our lower bound in the body of the paper (Theorem \ref{thm:eh}) is actually a bit more general than the one stated in Theorem \ref{thm:introeh} above.
Although it is popular to focus on hopsets of size $O(n)$, one can also ask about hopsets of size $O(p)$ for any parameter $p \in [1, n^2]$.
The folklore sampling algorithm extends to any $p$, by adjusting the size of the sampled vertex set to $p^{1/2}$.
Our generalized lower bound establishes that the diameter bound from folklore sampling is near-optimal in the entire range of parameters.

We now turn to shortcut sets.
Our new lower bound is the following:
\begin{theorem} [New]
There are $n$-node directed input graphs on which any shortcut set on $O(n)$ edges reduces diameter to $\Omegaish(n^{1/4})$.
\end{theorem}

This improves over the previous lower bound of $\Omega(n^{1/6})$ by Huang and Pettie \cite{HP18}, but a polynomial gap to the upper bound of $\Oish(n^{1/3})$ still remains \cite{KP22a}.
It is an interesting open question is to narrow this gap further.
We note that every $(1+\eps)$ hopset is also a shortcut set, and so this lower bound extends to $(1+\eps)$ hopsets as well.\footnote{Note that, since the shortcut set lower bound is only for directed graphs, the lower bound only extends to $(1+\eps)$ hopsets for \emph{directed} input graphs.  For $(1+\eps)$ hopsets in undirected (but possibly weighted) input graphs, Elkin, Gitlitz, and Neiman proved that much better diameter can be achieved \cite{EGN19}.}
In the body of the paper (Theorem \ref{thm:ss}) we prove a more general theorem giving improved lower bounds against $O(p)$ size hopsets, but this time the parameter range in which our extended theorem is nontrivial is only $p \in [1, n^{5/4}]$.

\subsection{Other Related Work and Open Questions}

This work is only concerned with the \emph{existential} bounds that can be achieved for shortcut sets and hopsets.
Some prior work in the area has also focused on constructions that are efficient in the appropriate computational model.
This was perhaps most famously achieved by Fineman \cite{Fineman19}, whose breakthrough algorithm for parallel reachability was centered around a new shortcut set construction.
His construction reduced diameter to $\Oish(n^{2/3})$ using $O(n)$ edges; this is a worse diameter bound than the one achieved by folklore sampling, but crucially its work-efficiency was much better than folklore.
This was later improved by Jambulapati, Liu, and Sidford \cite{JLS19}, who achieved roughly the diameter bound from folklore sampling with work-efficiency comparable to Fineman \cite{Fineman19}.
Relatedly, another work of Kogan and Parter \cite{KP22b} gave a construction improving the (centralized) construction time of their shortcut set algorithm.

This work focuses on hopsets for \emph{weighted} graphs, but hopsets can be studied for unweighted graphs as well.
Specifically for $(1+\eps)$ hopsets in undirected unweighted graphs, it is known that far better diameter bounds are achievable.
In particular, following preliminary constructions in \cite{KS97, SS99, Cohen00}, constructions of Huang and Pettie \cite{HP19} and Elkin and Neiman \cite{EN19} showed that one can reduce diameter to $n^{o(1)}$ using $O(n)$ hop edges.
These papers actually provide a more fine-grained tradeoff between hopset size, $\eps$, and diameter bound, which is shown to be essentially tight in \cite{ABP17}.
Hopsets for undirected unweighted graphs with larger stretch values were studied in \cite{BP20}, and a unification of the various hopset constructions in this setting was developed in \cite{NS22}.

This work focuses on shortcut/hopset bounds in the setting where the edge budget is $O(n)$, or more generally some function of $n$.
These objects are also sometimes studied in a related setting where the edge budget is $O(m)$, where $m$ is the number of edges in the input graph.
Lower bounds in this setting were achieved in \cite{Hesse03, HP18, LVWX22}; most recently, Lu, Williams, Wein, and Xu showed graphs where any $O(m)$-size shortcut set reduces diameter to $\Omega(n^{1/8})$.
On the upper bounds side, it is easy to get upper bounds as functions of both $m$ and $n$ -- for example, folklore sampling with $|S|=m^{1/2}$ sampled nodes yields a diameter bound of $\Oish(n/m^{1/2})$, and the construction by Kogan and Parter \cite{KP22a} implies a bound of $\Oish(n^{2/3} / m^{1/3})$ for $O(m)$-size shortcut sets.
However, it is an interesting open problem to construct $O(m)$-size shortcut/hopsets with nontrivial diameter upper bounds that depend only on $n$.
By ``nontrivial,'' we mean that one can always assume without loss of generality that the input graph is connected, and thus $m = \Omega(n)$, and so a construction of $O(n)$-edge shortcut/hopsets is always valid in the $O(m)$-setting.
A nontrivial construction is one that beats that bound.
\begin{oq}
Prove that, for any $m$, every $n$-node, $m$-edge directed graph has an $O(m)$-edge shortcut/hopset that reduces diameter to $O(n^c)$, where $c$ is a constant strictly less than the one that can currently be achieved for $O(n)$-edge shortcut/hopsets.
\end{oq}

Finally, we remark that the setting of exact hopsets in unweighted graphs seems to be unexplored.
Although unpublished to our knowledge, one can obtain a lower bound by applying a standard analysis in \cite{HP18, KP22, BHT22} to the $n$-node undirected unweighted distance preserver lower bound graphs constructed by Coppersmith and Elkin \cite{CE06}.
This would imply that any $O(n)$-size exact hopset on the Coppersmith-Elkin graphs would reduce diameter to $\Omega(n^{1/5})$.
Our new hopset lower bound can also be interpreted as an improved lower bound for this setting:\footnote{This corollary is not immediate from the discussion so far: since our shortcut set lower bound is directed, it is not clear that it would imply a lower bound against \emph{undirected} unweighted exact hopsets.  This holds specifically because our shortcut set lower bound construction happens to be layered.}
\begin{corollary}
There are $n$-node undirected unweighted input graphs on which any exact hopset on $O(n)$ edges reduces diameter to $\Omegaish(n^{1/4})$.
\end{corollary}

But on the upper bounds side, folklore sampling remains the best known algorithm, and it only reduces diameter to $\Oish(n^{1/2})$.
It would be interesting to narrow this gap, and in particular to confirm or refute whether folklore sampling is near-optimal. 
\begin{oq}
Is the folklore algorithm near-optimal for exact hopsets in unweighted graphs?
Or, alternately, does every $n$-node unweighted graph have an $O(n)$-size exact hopset that reduces diameter to $O(n^{1/2 - c})$, for some $c>0$?
\end{oq}

%% file: overview.tex
\section{Technical Overview}
\label{sec:prob_approach}
\subsection{Recap of Prior Work}

In order to explain the strategy used for our improved lower bounds, it will be helpful to first recall the construction of Huang and Pettie \cite{HP18} for shortcut set lower bounds, and the construction of Kogan and Parter \cite{KP22} for hopset lower bounds.
The Huang-Pettie shortcut set lower bound is a construction of a directed graph $G$ and a set of paths $\Pi$ with the following properties:
\begin{enumerate}
\item Each path in $\Pi$ is the unique path in $G$ between its endpoints
\item \label{cnd:ed} The paths in $\Pi$ are pairwise edge-disjoint
\item There are $|\Pi| = cn$ paths, where $c$ is a constant that can be selected as large as we want
\item Subject to the above constraints, we want to make the paths in $\Pi$ as long as possible.
In particular, in \cite{HP18}, every path in $\Pi$ has the same length $\ell = \Theta(n^{1/6})$.
\end{enumerate}

Let us see why these properties imply a lower bound.
Each time we add an edge $(u, v)$ to our hopset $H$, by uniqueness and edge-disjointness of paths in $\Pi$, there can be at most one path $\pi \in \Pi$ where the distance between its endpoints decreases due to $(u, v)$.
Thus, if we build a shortcut set of size $|H| = |\Pi|-1$, then for at least one path $\pi \in \Pi$ the distance between its endpoints is the same in $G$ as in $G \cup H$.
Thus, the final diameter of the graph is at least $\ell = \Theta(n^{1/6})$, giving the lower bound.

The Kogan-Parter exact hopset lower bound is similar, except that each path $\pi \in \Pi$ is only required to be a unique \emph{shortest} path between its endpoints in the weighted graph $G$.
This is a more relaxed constraint than requiring $\pi$ to be the unique path of any kind, and this additional freedom in the construction lets us improve the path lengths to $\ell = \Theta(n^{1/3})$.
Besides that, the argument is identical.





\subsection{Allowing Paths to Overlap}

The change in our construction is a relaxation of item (\ref{cnd:ed}); that is, the paths in our constructions are not pairwise edge-disjoint.
This has appeared in prior work only in a rather weak form: all of the lower bounds against $O(m)$-size shortcut sets use paths that may intersect pairwise on a single edge \cite{Hesse03, HP18, LVWX22}.
These constructions begin with a system of paths as above, and then apply a tool called the \emph{alternation product} which introduces path overlap.
However, the alternation product is not a useful tool for our purposes, and it does not appear in this paper at all.
The alternation product \emph{harms} the path lengths of the construction (relative to its number of nodes), in exchange for also reducing the number of edges $m$ relative to the number of paths in the construction. This is useful for constructing lower bounds against $O(m)$-size shortcut/hopsets, but is not helpful for our goal of constructing lower bounds against $O(n)$-size shortcut/hopsets.


In our construction, paths that may intersect pairwise on \emph{polynomially} many edges.
This property arises from an entirely different technique, and for an entirely different reason: our goal is to use this overlap to get \emph{improved} path lengths.
Let us first observe why we can tolerate some path overlap while maintaining correctness of the lower bound.
Suppose our shortcut set has a budget of $cn$ edges, and we construct a graph $G$ and a set of $|\Pi| = 2cn$ paths, where each path has length $\ell$ and each path is the unique path between its endpoints.
Let $P$ be the set of node pairs that are the endpoints of paths in $\Pi$, and consider the following potential function over shortcut sets $H$, which simply sums distances over critical pairs:
$$\phi\left( H \right) := \sum \limits_{(s, t) \in P} \dist_{G \cup H}(s, t)$$
Initially, we have $\phi(\emptyset) = |\Pi| \cdot \ell = 2cn\ell$.
Then we add edges to our shortcut set one at a time, gradually reducing $\phi$.
How much could any given shortcut edge $(u, v) \in H$ reduce $\phi$?
Clearly it could reduce by $\ell-1$, in the case where $(u, v)$ are a pair in $P$, since this reduces $\dist(u, v)$ from $\ell$ to $1$.
This is acceptable: if \emph{all} edges reduce $\phi$ by at most $\ell-1$, then the final potential would be
$$\phi(H) \ge \phi(0) - (\ell-1)cn = 2\ell cn - (\ell-1)cn = (\ell+1)cn.$$
Thus, over the $|P| = 2cn$ critical pairs, the average distance in $G \cup H$ is $\Theta(\ell)$, and so the lower bound still holds.

So the only overlap constraint we need to enforce in our lower bound is that no shortcut edge can reduce the potential $\phi$ by more than $\ell$.
This is a much more forgiving constraint than edge-disjoint paths.
For example, for two internal nodes $u, v$ with $\dist_G(u, v) = \ell/2$, we could allow two different paths to coincide on a $u \leadsto v$ subpath: adding $(u, v)$ to the shortcut set would then reduce $\phi$ by only $2 \cdot (\ell/2-1) = \ell-2$.
In general, for two nodes at distance $\ell/x$, we can safely allow $x$ paths to coincide on the subpath between these nodes while maintaining correctness of the lower bound.





  

\subsection{Constructing Overlapping Paths}

The previous part explains why we are \emph{allowed} overlapping paths, but it is still not clear how to leverage this freedom into an improved lower bound construction.
This is where our technical contribution lies.
It is again a bit easier to explain the new idea in the context of shortcut sets, but the intuition is essentially the same in the context of hopsets.

Let us return to the previous lower bound constructions.
For the shortcut lower bounds of \cite{HP18}, one constructs an $(\ell+1)$-layered directed graph where the nodes in each layer are a copy of a grid within $\zz^2$.
The next step is to construct a set of convex vectors $C$.
A key perspective shift in this paper is that we think of the vectors in $C$ as playing two \emph{independent} roles in this previous construction:
\begin{itemize}
\item They play the role of \emph{edge vectors}: we include an edge from a node $u$ in layer $i$ to a node $v$ in layer $i+1$ iff the difference between the grid points $u, v$ is a vector $v-u=\vec{c} \in C$.

\item They also play the role of \emph{(objective) direction vectors}.
The paths $\pi \in \Pi$ are indexed by a node $u$ in the first layer and a vector $\vec{c} \in C$, and we generate $\pi_{u, \vec{c}}$ by using $u$ as its first node, and then iteratively selecting its node in the next layer by adding $\vec{c}$ to the node in the previous layer.
(A technical detail is that only choices of $(u, \vec{c})$ are allowed that reach the last layer without the path falling off the side of the grid.)
Notice the argument for path uniqueness: using $\vec{c}$ as an objective direction, due to the convexity of the vectors in $C$, the edge vector $\vec{c}$ itself is the one that maximizes progress in the direction $\vec{c}$.
Thus, no alternate path beginning at $u$ can reach the grid point $u + \ell \vec{c}$ within $\ell$ steps, since it necessarily makes less progress in the direction of $\vec{c}$, implying path uniqueness.
\end{itemize}

Our constructions disentangle these two uses of the vector set $C$: we depart from prior work by explicitly using a separate edge vector set $C$ and direction vector set $D$.
These vector sets crucially do \emph{not} have the same size: instead we will have $|D| \gg |C|$, and this difference allows for a technical optimization of parameters leading to improved lower bounds.
Roughly, we can choose $D$ large enough to have $|\Pi| = \Theta(n)$, while also allowing $|C| \ll n$.
This smaller size $|C|$ can be achieved using shorter convex vectors, which in turn lets us pack more layers into the construction without worrying about paths falling off the sides of the grid before reaching the final layer.

We use the following generalized process for iteratively generating critical paths.
Each path $\pi \in \Pi$ is indexed by a node $u$ in the first layer and a direction vector $\vec{d} \in D$, and at each layer, we greedily select the edge vector $\vec{c} \in C$ that maximizes progress in the objective direction $\vec{d}$.
Since $|D| \gg |C|$, by necessity many different objective directions will all select the same edge vector at each layer.
This leads to overlapping paths  discussed above, but more technical ingredients are still needed to ensure that paths don't overlap too much.
We explain these next.

\subsection{Symmetry Breaking}

There is an important problem with the construction sketched so far.
Consider two nearby direction vectors $\vec{d_1}, \vec{d_2} \in D$, which have the same optimizing edge vector $\vec{c}$.
Then for any start node $u$, the paths $\pi_{u, \vec{d_1}}, \pi_{u, \vec{d_2}}$ will simply select the same edge vector $\vec{c}$ at each layer, and these two paths will entirely coincide.
In other words, the ``extra paths'' bought by using $|D| \gg |C|$ are actually just identical copies of a much smaller set of paths, which is not interesting or useful.

We therefore need to somehow break the symmetry between paths that use nearby direction vectors, getting them to eventually choose different edge vectors at some layer to split apart.
This is where our lower bound constructions diverge; we will need to use two different symmetry-breaking strategies for shortcut and hopset lower bounds.

\paragraph{Hopset Lower Bounds and $\eps$-Shifting.}

In our hopset lower bound construction $G$, like \cite{KP22}, our vertices can be interpreted as points in $\mathbb{R}^2$.
More specifically, they initially form a square grid within the integer lattice $\mathbb{Z}^2$, and the columns of this grid act as layers of $G$.
Our edges initially have the form $e = (v, v + \Vec{c})$ for edge vectors $\Vec{c} \in C$; edge vectors always have first coordinate $1$, so that they go from one layer to the next.
Initially, the weight of an edge is the Euclidean distance between its endpoints.

Our symmetry-breaking step is a random operation where for each layer $i$ we choose a random variable $\eps_i$ sampled uniformly from the interval $(0, 1)$, and we shift the $i^{th}$ layer upwards so that its nodes are offset $\eps_i$ higher than the nodes in the previous layer.
The shifts therefore compound across the layers.
See Figure \ref{fig:vertices} for a picture.

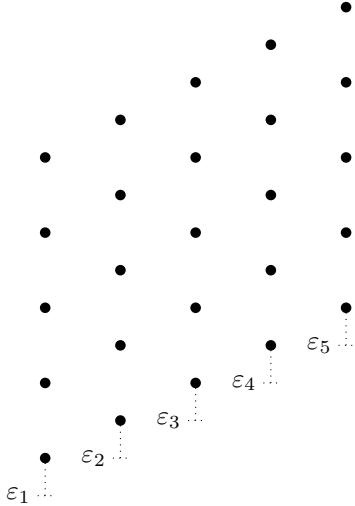
\begin{figure} [h]
\begin{center}
\begin{tikzpicture}
    \pgfmathsetmacro\rows{5} 
    \pgfmathsetmacro\cols{5} 
    \pgfmathsetmacro\dist{1} 
    \pgfmathsetmacro\shift{0.5} 

    \foreach \i in {1,...,\cols} {
        \foreach \j in {1,...,\rows} {
            \fill ({\i*\dist}, {\j*\dist + (\i-1)*\shift}) circle (2pt);
        }
    }

    
    \foreach \i [evaluate=\i as \prevI using int(\i-1)] in {1,...,\cols} {
        \draw [dotted] ({\i*\dist}, {\i*\shift}) -- ({\i*\dist}, {(\i+1)*\shift});
        \draw [dotted] ({\i*\dist - 0.1}, {\i*\shift}) -- ({\i*\dist + 0.1}, {\i*\shift});
        
        \path ({\i*\dist}, {\i*\shift}) -- node[left, midway, xshift=-2pt] {$\varepsilon_{\i}$} ({\i*\dist}, {\i*\shift});
    }
\end{tikzpicture}
\end{center}
\caption{\label{fig:vertices} Vertex set of the graph $G$ used for our lower bounds against exact hopsets.  Each parameter $\eps_i$ is the amount the $i^{th}$ column is shifted upwards in the plane, relative to the previous column; the $\{\eps_i\}$ values are chosen uniformly and independently from the interval $(0, 1)$.}
\end{figure}


Our $\eps$-shifting strategy does not affect the edge set of $G$, nor does it affect the set of direction vectors in any way, but it does affect the Euclidean distance between nodes in adjacent layers, and hence it changes the edge weights.
It achieves symmetry-breaking for roughly the following reason.
In our greedy generation of paths, a path with direction vector $\vec{d}$ will use the closest edge vector $\vec{c}$ at each level.
If two paths $\pi_1$ and $\pi_2$ with direction vectors $\Vec{d}_j, \Vec{d}_k \in D$ intersect at a node $v$ in the $i$th layer of $G$, then there will be an interval $(a, b) \subseteq (0, 1)$ such that if $\varepsilon_{i+1}$ lands in $(a, b)$, then $\pi_1, \pi_2$ have different closest edge vectors after shifting.
Thus, in this event, the paths $\pi_1$ and $\pi_2$ split apart at $v$ and never reconverge (this is formalized in Lemma \ref{lem:node_split}).
The size of the interval $(a, b)$, and hence the probability that it gets hit by $\eps_{i+1}$, is proportional to the distance between $\vec{d}_j$ and $\vec{d}_k$.
The effect is that paths generated by nearby direction vectors tend to intersect on long subpaths, while paths generated by far apart direction vectors intersect on shorter subpaths or perhaps just a single node, but with high probability all pairs of paths split apart eventually.
Paths generated by \emph{the same} direction vector remain parallel, and do not intersect at all.

There is a technical detail remaining: we still need to prove that each critical path is a unique shortest path between its endpoints.
In \cite{KP22}, the critical paths correspond to lines in Euclidean space, and since edge weights correspond to Euclidean distances, the analogous unique shortest paths property follows instantly from the geometry of $\rr^2$.
Since our critical paths are generated by a more involved process, it is much more technical to prove that they are unique shortest paths.
Proposition \ref{prop:squared_prop} contains the optimization lemma that needs to hold for our process to generate unique shortest paths, and to push it through, it turns out that we essentially need the \emph{derivative} of edge weights to be proportional to Euclidean distances.
We therefore differ again from \cite{KP22} by squaring all of our edge weights, meaning that our graph metric is ultimately quite different from $\rr^2$.
See Section \ref{subsec:uniqueshortest} for additional details.




\paragraph{Shortcut Set Lower Bounds and Edge Vector Subsampling.}

Shortcut set lower bounds are unweighted, and this makes the technique of $\eps$-sampling essentially useless in this setting, since it only affects edge weights in the construction and it does not change the edge \emph{set}.
For shortcut sets, we need an entirely different symmetry-breaking strategy that actually changes the edge set from layer to layer.

Our starting graph is similar to the one used by Huang and Pettie \cite{HP18}, mentioned earlier.
Each layer of the graph is an independent copy of a square grid in $\zz^2$.
We generate a large convex set $W \subseteq \rr^2$; initially, $W$ plays the role of both edge vectors $C$ and direction vectors $D$.
However, for the sake of symmetry-breaking, we do not put edges between \emph{all} nodes in adjacent layers whose difference is a vector in $W$.
Instead, at each layer $i$ we randomly sample exactly two adjacent vectors $\Vec{c}_{\lambda_i}, \vec{c}_{\lambda_i + 1} \in W$, and we use only these two edge vectors to generate edges to the next layer.
This is depicted in Figure \ref{fig:shortcutsb}.

\usetikzlibrary{3d}
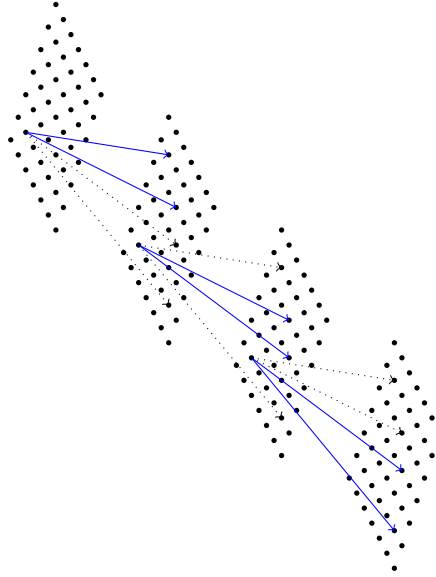
\begin{figure}[h]
\begin{center}
\begin{tikzpicture}[scale=1, x={(0.1cm, 0.3cm)}, y={(0cm, 0.1cm)}, z={(-0.1cm, 0.1cm)}]
  \def\gridsize{6} 
  \def\perspective{1} 
  \def\layercount{3} 
  \def\layergap{15} 

  \foreach \layer in {0, 1, ..., \layercount} {
    \foreach \x in {0, 1, ..., \gridsize} {
      \foreach \y in {0, 1, ..., \gridsize} {
        \pgfmathsetmacro{\z}{\perspective * \y + \layer * \layergap}
        \fill (\x, \y, \z) circle (1pt);

      }
    }
  }

  \draw [dotted, ->] (1, 5, {\perspective * 5 + 3 * \layergap}) -- (1, 1, {\perspective * 1 + 2 * \layergap});
  \draw [blue, ->] (1, 5, {\perspective * 5 + 3 * \layergap}) -- (5, 5, {\perspective * 5 + 2 * \layergap});
  \draw [blue, ->] (1, 5, {\perspective * 5 + 3 * \layergap}) -- (4, 3, {\perspective * 3 + 2 * \layergap});
  \draw [dotted, ->] (1, 5, {\perspective * 5 + 3 * \layergap}) -- (3, 2, {\perspective * 2 + 2 * \layergap});

  \draw [dotted, ->] (1, 5, {\perspective * 5 + 2 * \layergap}) -- (1, 1, {\perspective * 1 + 1 * \layergap});
  \draw [dotted, ->] (1, 5, {\perspective * 5 + 2 * \layergap}) -- (5, 5, {\perspective * 5 + 1 * \layergap});
  \draw [blue, ->] (1, 5, {\perspective * 5 + 2 * \layergap}) -- (4, 3, {\perspective * 3 + 1 * \layergap});
  \draw [blue, ->] (1, 5, {\perspective * 5 + 2 * \layergap}) -- (3, 2, {\perspective * 2 + 1 * \layergap});

  \draw [blue, ->] (1, 5, {\perspective * 5 + 1 * \layergap}) -- (1, 1, {\perspective * 1 + 0 * \layergap});
  \draw [dotted, ->] (1, 5, {\perspective * 5 + 1 * \layergap}) -- (5, 5, {\perspective * 5 + 0 * \layergap});
  \draw [dotted, ->] (1, 5, {\perspective * 5 + 1 * \layergap}) -- (4, 3, {\perspective * 3 + 0 * \layergap});
  \draw [blue, ->] (1, 5, {\perspective * 5 + 1 * \layergap}) -- (3, 2, {\perspective * 2 + 0 * \layergap});

\end{tikzpicture}
\end{center}
\caption{\label{fig:shortcutsb} Our symmetry-breaking strategy for shortcut set lower bounds starts with a large set of convex vectors, but independently subsamples adjacent pairs of convex vectors to generate the edges between adjacent layers.  In this picture, there are $4$ edge vectors and $4$ layers, but only two of the edge vectors (in blue) are sampled and available between any given pair of adjacent layers.  For clarity, we have only drawn the edges leaving one particular node in each layer.}
\end{figure}

The fact that the sampled vectors are \emph{adjacent}, and hence typically close together, allows for a key optimization in the construction.
The rate at which paths drift apart from each other is much slower than in \cite{HP18}, even when they are generated by very different direction vectors.
This allows us to apply a carefully-chosen translation of the grid from layer to layer, in order to keep all of the paths contained in the grid.
This in turn lets us pack many more layers into the construction while still ensuring that all of our paths stay within the confines of the grid.

As before, paths $\pi \in \Pi$ are generated greedily: for direction vector $\Vec{d} \in D$, an associated path $\pi$ will traverse the sampled edge vector in each layer that maximizes progress in the objective direction $\Vec{d}$. 
For two paths $\pi_1, \pi_2$ with direction vectors $\Vec{d}_j, \Vec{d}_k \in D$ and $v \in \pi_1 \cap \pi_2$, these paths split at $v$ in the event that the sampled edge vectors $\Vec{c}_{\lambda_i}, \vec{c}_{\lambda_i + 1} \in W$ lie \textit{between} $\Vec{d}_j$ and $\Vec{d}_k$ in $D$.
This again leads to behavior where paths generated by nearby direction vectors tend to coincide on long subpaths, while paths generated by far apart direction vectors have smaller intersections, but with high probability all pairs of paths split apart eventually.
Paths with the same direction vector again remain parallel.
This is formalized in Lemma \ref{lem:node_split_ss}.


%% file: hopset.tex
\section{Preliminaries}


We use the following notations:
\begin{itemize}
\item For a path $\pi$, we use $|\pi|$ to denote the number of \textbf{nodes} in $\pi$.
This is $1$ different from the (unweighted) length of $\pi$.
In weighted graphs, we write $w(\pi)$ for the sum of edge weights in $\pi$.

\item We write $\dist_G(s, t)$ for the shortest path distance from node $s$ to node $t$ in graph $G$ (counting edge weights, if $G$ is weighted).
We write $\hopdist_G(s, t)$ for the least number of edges contained in any $s \leadsto t$ shortest path.

\item We use $\langle \cdot, \cdot \rangle$ to denote the standard Euclidean inner product.
\end{itemize}

\section{Exact Hopsets}






In this section we will prove the following theorem.

\begin{theorem} \label{thm:ehbody}
For any parameter $ p \in [1, n^2]$, there exists an $n$-node weighted undirected graph $G = (V, E, w)$ such that for any exact hopset $H$ of size $|H| \leq p$ where $p \in [1, n^2]$, the graph $G \cup H$ must have hopbound  $\Omega\left(\frac{n}{p^{1/2} \log^{1/2} n}\right)$.
\label{thm:eh}
\end{theorem}

\noindent 
We will prove this via a construction of the following type:
\begin{lemma}
\label{lem:construction_lemma}
For any $p \in [1, n^2]$, there is an infinite family of $n$-node undirected weighted graphs $G = (V, E, w)$ and sets $\Pi$ of $|\Pi|=p$ paths in $G$ with the following properties:
\begin{itemize}
\item $G$ has $\ell = \Theta\left( \frac{n}{p^{1/2} \log^{1/2} n} \right)$ layers, and each path in $\Pi$ starts in the first layer, ends in the last layer, and contains exactly one node in each layer.

\item Each path in $\Pi$ is the unique shortest path between its endpoints in $G$.

\item For any two nodes $u, v \in V$, there are at most $ \frac{\ell}{\hopdist_G(u, v)} $ paths in $\Pi$ that contain both $u$ and $v$.

\item Each node $v \in V$ lies on at most $O\left(\frac{\ell p}{n} \right)$ paths in $\Pi$. 
\end{itemize}
\end{lemma}

\noindent
Next we show how Lemma \ref{lem:construction_lemma} implies Theorem \ref{thm:ehbody}.

\subsection{Proving Theorem \ref{thm:ehbody} using Lemma \ref{lem:construction_lemma}}
\label{subsec:eh_thm}
Fix an $n$ and $p \in [1, n^2]$. Let $G = (V, E, w)$ be the graph in Lemma \ref{lem:construction_lemma} with associated set $\Pi$ of $|\Pi| = 2p$ paths in $G$. Let $H$ be an exact hopset of size $|H| \leq p$. Let $P \subseteq V \times V$ be the set of node pairs that are the endpoints of paths in $\Pi$. We define the following potential function over hopsets $H$, which simply sums hopdistances over critical pairs:
$$\phi\left( H \right) := \sum \limits_{(s, t) \in P} \hopdist_{G \cup H}(s, t)$$
Observe that by Lemma \ref{lem:construction_lemma}, we have $\phi(\emptyset) = \sum_{(s, t) \in P}\hopdist_G(s, t) =  |\Pi| \cdot (\ell - 1) = 2p(\ell - 1)$. 
Now fix a pair of nodes $(x, y) \in V \times V$, and let $\Pi' \subseteq \Pi$ be the set of paths $\pi \in \Pi$ such that $x, y \in \pi$. We make the following observations.
\begin{itemize}
    \item For all $(s, t) \in P$, if the unique shortest $s \leadsto t$-path $\pi \in \Pi$ in $G$ is not in $\Pi'$, then
    $$
    \hopdist_{G \cup \{(x, y)\}}(s, t) = \hopdist_{G}(s, t).
    $$
    \item For all $(s, t) \in P$, if the unique shortest $s \leadsto t$-path $\pi \in \Pi$ in $G$ is in $\Pi'$, then
    $$
     \hopdist_{G}(s, t) - \hopdist_{G \cup \{(x, y)\}}(s, t) \leq \hopdist_G(x, y).
    $$
\end{itemize}
Then by Lemma \ref{lem:construction_lemma}, $\phi(\emptyset) - \phi(\{(x, y)\}) \leq |\Pi'| \cdot \hopdist_G(x, y) \leq \ell$. 
We obtain the following sequence of inequalities:
$$
    \phi(\emptyset) - \phi(H)  \leq \sum_{(x, y) \in H} \left(\phi(\emptyset) - \phi(\{(x, y)\})\right)  \leq \ell \cdot |H| \leq \ell p.
$$
Rearranging, we find that
$$
\phi(H) \geq \phi(\emptyset) - \ell p \geq 2p(\ell - 1) - \ell p = (\ell - 2)p.
$$
Thus, over the $|P| = 2p$ pairs of path endpoints in $P$, the average hopdistance in $G \cup H$ is $\Theta(\ell)$, and so  there must be a pair $(s, t) \in P$ such that $\hopdist_{G \cup H}(s, t) = \Theta(\ell)$ by the pigeonhole principle.

\subsection{Constructing $G$}
Our goal is now to prove Lemma \ref{lem:construction_lemma}. Let $n$ be a sufficiently large positive integer, and let $p \in [n, n^2]$.\footnote{We will handle the case where $p \in [1, n]$ later.} 
For simplicity of presentation, we will frequently ignore issues related to non-integrality of expressions that arise in our construction; these issues affect our bounds only by lower-order terms. 
Initially, all edges $(u, v)$ in $G$ will be directed from $u$ to $v$; we will convert $G$ into an undirected graph in the final step of our construction. 


\paragraph{Vertex Set $V$.}
\begin{itemize}
    \item  Let $\ell$ be a positive integer parameter of the construction to be specified later.  Our graph $G$ will have $\ell$ layers $L_1, \dots, L_{\ell}$, and each layer will have $n/\ell$ nodes, ordered from $1$ to $n/\ell$. Initially, we will label the $j$th node in layer $L_i$ with tuple $(i, j)$. We will interpret the node labeled $(i, j)$ as a point in $\mathbb{R}^2$ with integer coordinates. These $n$ nodes arranged in $\ell$ layers will be the node set $V$ of graph $G$. 
    \item We now perform the following random operation on the node labels of $V$. For each layer $L_i$, $i \in [1, \ell]$, uniformly sample a random real number in the interval $(0, 1)$ and call it $\varepsilon_i$. Now for each node in layer $L_i$ of $G$ labeled $(i, j)$, relabel this node with the label
    $$
    \left(i, j+\sum_{k=1}^j \varepsilon_k\right).
    $$
    Again, we interpret the resulting labels for nodes in $V$ as points in $\mathbb{R}^2$.
    In a slight abuse of notation, we will treat $v \in V$  as either a node in $G$ or a point in $\mathbb{R}^2$, depending on the context. 
    Less formally: for each layer $i$, this step shifts the nodes in layer $i$ vertically upwards to be $\eps_i$ higher than the previous layer (and thus, these vertical shifts compound across the layers). See Section \ref{sec:prob_approach} for intuition on this design choice.

\end{itemize}


\paragraph{Edge Set $E$.} 
\begin{itemize}
    \item All the edges in $E$ will be between consecutive layers $L_i, L_{i+1}$ of $G$. We will let $E_i$ denote the set of edges in $G$ between layers $L_i$ and $L_{i+1}$. 
    \item Just as our nodes in $V$ correspond to points in $\mathbb{R}^2$, we can interpret the edges $E$ in $G$ as vectors in $\mathbb{R}^2$. In particular, for every edge $e = (v_1, v_2) \in E$, we identify $e$ with  the corresponding vector $\Vec{u}_e := v_2 - v_1$. Note that since all edges in $E$ are between adjacent layers $L_i$ and $L_{i+1}$, the first coordinate of $\Vec{u}_e$ is $1$ for all $e \in E$. We will use $u_e$ to denote the 2nd coordinate of $\Vec{u}_e$, i.e., for all $e \in E$, we write $\Vec{u}_e = (1, u_e)$.
    \item We begin our construction of $E$ by defining the following set $C$ of vectors: 
    $$
    C := \left\{ (1, x) \mid x \in \left[0, \frac{n}{4\ell^2}\right] \right\}.
    $$
    We will refer to the vectors in $C$ as \emph{edge vectors}.
    
    \item For each $i \in [1, \ell - 1]$, let 
    $$
    C_i := \{ \Vec{c} + (0, \varepsilon_{i+1}) \}_{\Vec{c} \in C}.
    $$
    Intuitively: we want the edge vectors in $C$ to point between nodes in adjacent layers, and due to the random vertical shifts between layers applied to the nodes, we need to apply a similar shift to $C$ at each layer to adjust for this. 
    
    \item For each $v \in L_i$ and edge vector $\Vec{c} \in C_i$, if $v + \Vec{c} \in V$, then add edge $(v, v+ \Vec{c})$ to $E_i$. After adding these edges to $E_i$, we will have that
    $$
    C_i = \{\Vec{u}_e \mid e \in E_i\}.
    $$
    Additionally, note that the case $v + \Vec{c} \not \in V$ only occurs if $v + \Vec{c} = (i+1, j)$ for some $j$ that is higher than any point in the $(i+1)$st layer; that is, $j > \frac{n}{\ell} + \sum_{k = 1}^{i+1} \varepsilon_k$. 
    
    \item For each $e \in E$, if $\Vec{u}_e = (1, u_e)$, then we assign edge $e$ the weight $w(e) := u_e^2$. 
\end{itemize}
This completes the construction of our graph $G = (V, E, w)$.

\subsection{Direction Vectors, Critical Pairs, and Critical Paths}
\label{subsec:critical_paths}
Our next step is to generate a set of \emph{critical pairs} $P \subseteq V \times V$, as well as a set of \emph{critical paths} $\Pi$.
Specifically, there will be one critical path $\pi_{s,t} \in \Pi$ going between each critical pair $(s, t) \in P$, and we will show that $\pi_{s,t}$ is the unique shortest (weighted) $s \leadsto t$ path in $G$.
 We will identify our critical pairs and paths by first constructing a set of vectors $D \subseteq \mathbb{R}^2$ that we call \textit{direction vectors}, which we define next.
\paragraph{Direction Vectors $D$.}
\begin{itemize}
    \item Let $q \in \mathbb{Z}_{\geq 0}$ be a sufficiently large
    integer parameter to be specified later. The size of $q$ will roughly correspond to the maximum number of edges shared between any two critical paths in $G$.
    
    \item We choose our set of direction vectors $D$ to be \footnote{Note that if $(1, x) \in D$, then $x \in \left[1, \frac{n}{4\ell^2}  \right]$. However, if $(1, x) \in C$, then $x \in \left[0, \frac{n}{4\ell^2}\right]$. This $+1$ gap between $C$ and $D$ is needed to accommodate the $\varepsilon$-shifting operation used to obtain $C_i$, and is relevant in the proof of Lemma  \ref{lem:node_split}. }
    $$
    D := \left\{ \left(1, \text{ } x + \frac{y}{q} \right) \bigl\vert x \in \left[ 1, \text{ } \frac{n}{4\ell^2} - 1 \right] \text{ and } y \in [0, q]  \right\}. 
    $$
    Note that there are $q+1$ direction vectors between adjacent vectors $(1, x), (1, x+1) \in C$ for $x\neq 0$. Additionally, adjacent direction vectors in $D$ differ only by $1/q$ in their second coordinate.
\end{itemize}
\begin{prop}
With probability $1$, for every $i \in [1, \ell -1]$ and every direction vector $\Vec{d} = (1, d) \in D$, there is a unique vector $(1, c) \in C_i$ that minimizes $|c - d|$ over all choices of $(1, c) \in C_i$. 
\label{prop:distinct_obj_dir}
\end{prop}
\begin{proof}
There are only finitely many choices of $\varepsilon_{i+1} \in (0, 1)$ that result in there being two distinct vectors $(1, c_1), (1, c_2) \in C_i$ such that $|c_1 - d| = |c_2 - d|$. We conclude that the claim holds with probability $1$. 
\end{proof}
In the following we assume that this event holds, i.e., there is a unique minimizing vector in $C_i$ for all $\vec{d} \in D$.
Each of our critical paths $\pi$ in $\Pi$ will have an associated direction vector $\Vec{d} \in D$, and for all $i \in [1, \ell - 1]$, path $\pi$ will take an edge vector in $C_i$ that is closest to $\Vec{d}$ in the sense of Proposition \ref{prop:distinct_obj_dir} (see Section \ref{sec:prob_approach} for more intuition).


\paragraph{Critical Pairs $P$ and Critical Paths $\Pi$.}
\begin{itemize}
    \item We first define a set $S \subseteq L_1$ containing half of the nodes in the first layer $L_1$ of $G$:
$$
S := \left\{ (1, j + \varepsilon_1) \in L_1 \mid j \in \left[1, \frac{n}{2 \ell}\right] \right\}.
$$
We will choose our set of demand pairs $P$ so that $P \subseteq S \times L_{\ell}$.
For every node $s \in S$ and direction vector $\Vec{d} \in D$, we will create a critical pair $(s, t) \in S \times L_{\ell}$ and a corresponding critical path $\pi_{s, t}$ to add to $P$ and $\Pi$.

\item 
Let $v_1 \in S$, and let $\Vec{d} = (1, d) \in D$.
The associated path $\pi$ has start node $v_1$.
We iteratively grow $\pi$, layer-by-layer, as follows.
Suppose that currently $\pi = (v_1, \dots, v_i)$, for $i<\ell$, with each $v_i \in L_i$.
To determine the next node $v_{i+1} \in L_{i+1}$, let $E_i^{v_i} \subseteq E_i$ be the edges in $E_i$ incident to $v_i$, and let
$$
e_i := \text{argmin}_{e \in E_i^{v_i}}(|u_e - d|).
$$
By definition, $e_i$ is an edge whose first node is $v_i$; we define $v_{i+1} \in L_{i+1}$ to be the other node in $e_i$, and we append $v_{i+1}$ to $\pi$.



\item This completes our construction of $P$ and $\Pi$. Note that
\begin{itemize}
    \item we will show that the paths generated in this way have distinct endpoints (with high probability), and therefore $|P| = |S||D| \geq \frac{n^2q}{16\ell^3}$, and
    \item every path $\pi_{s, t} \in \Pi$ contains one node in each layer, and therefore its number of nodes is $|\pi_{s, t}| = \ell $. 
\end{itemize}
\end{itemize}

An important feature for correctness of our construction is that, when we iteratively generate paths, we never reach a point $v_{i}$ such that $v_i + C_i \not \subseteq L_{i+1}$ (i.e. $v_i + \Vec{c} \not \in L_{i+1}$ for some $\Vec{c} \in C_i$). 
This follows by straightforward counting, based on the maximum second coordinate used in our edge vectors $C$ and also on our choice of start nodes $S$ as only the ``lower half'' of the nodes in the first layer.
The following proposition expresses this correctness in a particular way, pointing out that for any node $v$ lying on a generated path $\pi$, none of the edges from $v$ to the following layer are omitted from the graph due to falling off the top of the grid with a too-high second coordinate.


\begin{prop}
\label{prop:path_edge}
Let $v \in L_i \cap \pi$ for some $i \in [1, \ell - 1]$ and $\pi \in \Pi$. Then $\{\Vec{u}_e \mid e \in E_i^v \} = C_i$.
\end{prop}
\begin{proof}
Let $v =: (i, j)$, and
let $(s, t)\in P$ be the endpoints of $\pi$. 
Since $s \in S$, we have $s = (1, s_2) \in \mathbb{R}^2$, where
$$ s_2 \leq \frac{n}{2\ell} + \varepsilon_1.$$ Moreover, since for all $e \in E_i$ the corresponding vector $\Vec{u}_e = (1, u_e)$ satisfies $\varepsilon_{i+1} \leq u_e \leq \frac{n}{4 \ell^2} + \varepsilon_{i+1}$, we have
$$ j \leq \frac{n}{2\ell} +  (i - 1) \cdot \frac{n}{4\ell^2} + \sum_{k=1}^i \varepsilon_k \leq \frac{3n}{4\ell} + \sum_{k=1}^i \varepsilon_k.$$
Then observe that for all $\Vec{c} = (1, c) \in C_i$, we have that $v + \Vec{c} = (i+1, j+c)$, where
$$j + c \leq \left(\frac{3n}{4\ell} + \sum_{k=1}^i \varepsilon_k\right) + \left( \frac{n}{4\ell^2} + \eps_{i+1} \right)  \leq \frac{n}{\ell} + \sum_{k=1}^{i+1} \varepsilon_k.$$
Thus $(i+1, j+c) \in L_{i+1}$, and so we have $(v, v+ \Vec{c}) \in E_i^v$ for all $\Vec{c} \in C_i$.
It follows that $\{\Vec{u}_e \mid e \in E_i^v\} = C_i$. 
\end{proof}

\subsection{Critical paths are unique shortest paths}
\label{subsec:uniqueshortest}
We now verify that graph $G$ and paths $\Pi$ have the unique shortest path property as stated in Lemma \ref{lem:construction_lemma}. 

\begin{lemma}[Unique shortest paths]
\label{lem:unique_paths}
With probability $1$, for every $(s, t) \in P$, path $\pi_{s, t} \in \Pi$ is a unique shortest (weighted) $s \leadsto t$-path in $G$.
\end{lemma}

\noindent
We begin with a technical proposition:


\begin{prop}
\label{prop:squared_prop}
    Let $b, x_1,  \dots, x_k \in \mathbb{R}$. Now consider $\hat{x}_1,  \dots, \hat{x}_k$ such that 
    \begin{itemize}
        \item $|\hat{x}_i - b| \leq |x_i - b|$ for all $i \in [1, k]$, and
        \item $\sum_{i=1}^k x_i = \sum_{i=1}^k \hat{x}_i$. 
    \end{itemize}
    Then 
    $$
    \sum_{i=1}^k x_i^2 \geq \sum_{i=1}^k \hat{x}_i^2,
    $$
    with equality only if  $|\hat{x}_i - b| = |x_i - b|$ for all $i \in [1, k]$.
\end{prop}

\begin{proof}
We will prove the equivalent statement $\sum_{i=1}^k (\hat{x}_i^2 - x_i^2) \leq 0$. Fix an $i \in [1, k]$.
First we will show that
$$
\hat{x}_i^2 - x_i^2  \leq 2b(\hat{x}_i - x_i).
$$
We  split our analysis into four cases:
\begin{itemize}
    \item \textbf{Case 1: $b \leq \hat{x}_i \leq x_i$.} In this case, $\hat{x}_i^2 - x_i^2 = (\hat{x}_i + x_i)(\hat{x}_i - x_i) \leq 2b(\hat{x}_i - x_i)$.
    \item \textbf{Case 2: $x_i \leq \hat{x}_i \leq b$.} In this case,  $\hat{x}_i^2 - x_i^2 = (\hat{x}_i + x_i)(\hat{x}_i - x_i) \leq 2b(\hat{x}_i - x_i)$.
    \item \textbf{Case 3: $\hat{x}_i \leq b \leq x_i$.} In this case, $\hat{x}_i^2 - x_i^2 = (\hat{x}_i + x_i)(\hat{x}_i - x_i)\leq 2b(\hat{x}_i - x_i)$, since $ b - \hat{x}_i \leq x_i - b$.
    \item \textbf{Case 4: $x_i \leq b \leq \hat{x}_i$.} In this case,  $\hat{x}_i^2 - x_i^2 = (\hat{x}_i + x_i)(\hat{x}_i - x_i) \leq 2b(\hat{x}_i - x_i)$, since $\hat{x}_i - b \leq b - x_i$.
\end{itemize}
Then
\begin{align*}
    \sum_{i=1}^k (\hat{x}_i^2 - x_i^2) & \leq 2b \sum_{i=1}^k (\hat{x}_i - x_i) = 0.
\end{align*}
This inequality is strict if $|\hat{x}_i - b| < |x_i - b|$ for some $i \in [1, k]$. 
\end{proof}

Using Proposition \ref{prop:squared_prop}, we can now prove Lemma \ref{lem:unique_paths}.

\begin{proof}[Proof of Lemma \ref{lem:unique_paths}]
Fix an $(s, t) \in P$, and let $(1, x) \in D$ be the direction vector  associated with $\pi_{s, t}$. Let $\hat{x}_1, \dots, \hat{x}_{\ell - 1} \in \mathbb{R}$ be  real numbers such that the $i$th edge of $\pi_{s, t}$ has the corresponding vector $(1, \hat{x}_i) \in C_i$ for $i \in [1, \ell - 1]$. 
Now consider an arbitrary $s \leadsto t$-path $\pi$ in $G$, where $\pi \neq \pi_{s, t}$. Since all edges in $G$ are directed from $L_i$ to $L_{i+1}$, it follows that $\pi$ has $\ell - 1$ edges and the $i$th edge of $\pi$ is in $E_i$. Let $x_1, \dots, x_{\ell - 1} \in \mathbb{R}$ be real numbers such that the $i$th edge of $\pi$ has the corresponding vector $(1, x_i) \in C_i$ for $i \in [1, \ell - 1]$. 
Now observe that since $\pi$ and $\pi_{s, t}$ are both $s \leadsto t$-paths, it follows that $$\sum_{i=1}^{\ell - 1} \hat{x}_i = \sum_{i=1}^{\ell - 1} x_i.$$ Additionally, by our construction of $\pi_{s, t}$, it follows that   
$$
|\hat{x}_i - x | \leq |x_i - x|
$$
for all $i \in [1, \ell - 1]$. In particular, since $\pi \neq \pi_{s, t}$, there must be some $j \in [1, \ell - 1]$ such that $\hat{x}_j \neq x_j$, and so by Proposition \ref{prop:distinct_obj_dir},  $|\hat{x}_j - x | < |x_j - x|$ with probability 1.  
Then by Proposition \ref{prop:squared_prop},
$$
w(\pi_{s, t}) = \sum_{e \in \pi_{s, t}}w(e) = \sum_{i=1}^{\ell - 1} \hat{x}_i^2 < \sum_{i=1}^{\ell - 1} x_i^2 = \sum_{e \in \pi} w(e) = w(\pi).
$$
Path $\pi_{s, t}$ is a unique shortest $s \leadsto t$-path in $G$, as desired. 
\end{proof}

\subsection{Critical Paths Intersection Properties}

Before finishing our proof of Lemma \ref{lem:construction_lemma}, we will need to establish several properties of the critical paths in $\Pi$. 

\begin{prop}
Let $\pi_1, \pi_2 \in \Pi$ be two critical paths with the same corresponding direction vector $\Vec{d} \in D$. Then $\pi_1 \cap \pi_2 = \emptyset$.    
\label{prop:dist_dir}
\end{prop}
\begin{proof}
    Let $v_i^j \in L_i$ denote the $i$th node of $\pi_j$, where $j \in \{1, 2\}$.  
    Note that since $\pi_1$ and $\pi_2$ share the same direction vector $\Vec{d}$, 
    the edges $(v_i^1, v_{i+1}^1)$ and $(v_i^2, v_{i+1}^2)$ have the same corresponding vector $\Vec{u}_i \in C_i$ for all $i \in [1, \ell - 1]$  by Proposition \ref{prop:path_edge}.
    By our construction of $\Pi$, for each node in the first layer $v \in L_1$, $v$ belongs to at most one path $\pi \in \Pi$ with direction vector $\Vec{d}$, so $v_1^1 \neq v_1^2$. 
    Then for all $i \in [1, \ell]$,
    $$
    v_i^1 = v_1^1 + \sum_{i=1}^{i-1} \Vec{u}_i \neq v_1^2 + \sum_{i=1}^{i-1} \Vec{u}_i = v_i^2. \qedhere
    $$
\end{proof}

Let $\pi_1, \pi_2 \in \Pi$ be two critical paths, and let $v \in V$ be a node in $G$. We say that paths $\pi_1$ and $\pi_2$ \textit{split} at $v$ if $v \in \pi_1 \cap \pi_2$ and 
the node following $v$ in $\pi_1$ is distinct from the node following $v$ in $\pi_2$, and we simply say that $\pi_1$ and $\pi_2$ \emph{split} if there exists some $v \in V$ such that they split at $v$.
Note that since $\pi_1, \pi_2 \in \Pi$ are unique shortest paths in $G$, paths $\pi_1$ and $\pi_2$ can split at most once.







\begin{lemma}
\label{lem:node_split}
Fix a node $v \in L_i$, where $i \in [1, \ell - 1]$, and let $\pi_1, \pi_2 \in \Pi$ be critical paths with direction vectors $(1, d_1), (1, d_2) \in D$ such that $v \in \pi_1$ and $v \in \pi_2$. Then paths $\pi_1$ and $\pi_2$ split at $v$ with probability at least $\min \left\{|d_1 - d_2|, 1\right\}$.\footnote{For the sake of completeness, let us be more precise here about the probability claim being made in this lemma.
Consider any two paths $\pi_1, \pi_2$, indexed by two start nodes and two direction vectors, and consider a node $v \in L_i$.
The event that we generate $\pi_1, \pi_2$ in such a way that $v \in \pi_1 \cap \pi_2$ depends only on the random choices of $\eps_1, \dots, \eps_i$. If $v \in \pi_1 \cap \pi_2$, then the event that $\pi_1, \pi_2$ split at $v$ depends only on the random choice of $\eps_{i+1}$.
The claim is: conditional on the event that $\{\eps_1, \dots, \eps_{i}\}$ are selected in such a way that $v \in \pi_1 \cap \pi_2$, the probability that $\eps_{i+1}$ is selected such that $\pi_1, \pi_2$ split at $v$ is at least $\min \left\{|d_1 - d_2|, 1\right\}$.}
\end{lemma}
\begin{proof}
By Proposition \ref{prop:dist_dir}, $d_1 \neq d_2$, and assume wlog that $d_1 < d_2$. 
Let $F$ be the event that the random variable $\varepsilon_{i+1}$ was sampled so that 
$$
(\mathbb{Z} + \varepsilon_{i+1} + 1/2) \cap (d_1, d_2) \neq \emptyset,
$$
(where $(d_1, d_2) \subseteq \mathbb{R}$ is the open interval with endpoints $d_1$ and $d_2$).
Our proof strategy is to show that $F$ implies that $\pi_1, \pi_2$ split at $v$, and then to show that $F$ occurs with the claimed probability.

\paragraph{$F$ implies that $\pi_1, \pi_2$ split at $v$.}

Assume that $F$ occurs.
By construction there is a nonnegative integer $c \in \mathbb{Z}$ such that $ c + \varepsilon_{i+1} + 1/2$ is in the interval $(d_1, d_2)$. Since $(d_1, d_2) \subseteq \left(1, \frac{n}{4\ell^2 } \right)$, it follows that vectors $(1, c + \varepsilon_{i+1}), (1, c + \varepsilon_{i+1} + 1) \in \mathbb{R}^2$ are in $C_i$, because $0 \leq c \leq \frac{n}{4\ell^2} - 1$. More generally, by our choice of sets $C$ and $D$ there are  vectors $(1, c_1), (1, c_2) \in C_i$ such that
$$
c_1 \leq d_1 \leq c_1 + 1  \text{ \qquad and \qquad } c_2 \leq d_2 \leq c_2 + 1.
$$
Now we claim that 
$$
\text{argmin}_{(1, x) \in C_i} |x - d_1| \neq \text{argmin}_{(1, x) \in C_i} |x - d_2|. 
$$
To see this, suppose for the sake of contradiction that there is a vector $(1, y) \in C_i$ such that $$(1, y) = \text{argmin}_{(1, x) \in C_i} |x - d_1| = \text{argmin}_{(1, x) \in C_i} |x - d_2|.$$ 
Then using our assumption that $d_1 < c + \varepsilon_{i+1} + 1/2 < d_2$, we obtain 
$$
y - d_1 \leq |y - d_1| \leq \min\{|c_1 - d_1|, |(c_1 + 1) - d_1|\} \leq 1/2 < (c + \varepsilon_{i+1} + 1) - d_1
$$
and
$$
d_2 - y \leq |y - d_2| \leq \min\{|c_2 - d_2|, |(c_2 + 1) - d_2|\} \leq 1/2 < d_2 - (c + \varepsilon_{i+1}).
$$
Together, these two sequences of inequalities imply that $c + \varepsilon_{i+1}< y < c + \varepsilon_{i+1} + 1$.  But this contradicts our assumption that $(1, y) \in C_i$, so we conclude that 
$$
\text{argmin}_{(1, x) \in C_i} |x - d_1| \neq \text{argmin}_{(1, x) \in C_i} |x - d_2|. 
$$
By Proposition \ref{prop:path_edge}, $\{\Vec{u}_e \mid e \in E_i^v\} = C_i$, so we have also shown that 
$$
\text{argmin}_{\{\Vec{u}_e \mid e \in E_i^v\}} |u_e - d_1| \neq \text{argmin}_{\{\Vec{u}_e \mid e \in E_i^v\}} |u_e - d_2|.
$$
Then $\pi_1$ and $\pi_2$ must split at $v$ by our construction of the critical paths in $\Pi$. 

\paragraph{$F$ happens with good probability.}

Since $\varepsilon_{i+1}$ is sampled uniformly at random from the interval $(0, 1)$,
it follows that $(\mathbb{Z} + \varepsilon_{i+1} + 1/2) \cap (d_1, d_2) \neq \emptyset$ with probability at least $\min\{d_2 -d_1, 1\}$.
\end{proof}

We will use Lemma \ref{lem:node_split} to prove the following two lemmas, which capture the key properties of our graph $G$. 

\begin{lemma}[Low path overlap]
Let $\pi_1, \pi_2 \in \Pi$ be critical paths with distinct associated direction vectors $(1, d_1), (1, d_2) \in D$. Then:
\begin{itemize}
\item If $|d_1 - d_2| < 1$, then with probability at least $1 - n^{-8}$, we have\footnote{Formally, we consider any two paths $\pi_1, \pi_2$ indexed by two start nodes and direction vectors.  When we iteratively generate these paths, the number of nodes in their intersection (possibly $0$) depends only on the random choices of $\eps_1, \dots, \eps_{\ell}$.  The probability claim in this lemma is with respect to these random choices.}
$|\pi_1 \cap \pi_2| \leq \frac{8 \log n}{ |d_1 - d_2|}.$

\item If $|d_1 - d_2| \ge 1$, then $|\pi_1 \cap \pi_2| \leq 1$ (deterministically).
\end{itemize}
\label{lem:path_overlap}
\end{lemma}

\begin{proof}
We begin with the first point; suppose $|d_1 - d_2| < 1$.
Suppose we iteratively generate $\pi_1, \pi_2$ one layer at a time.
Each time we choose a node $v$ that lies in both $\pi_1$ and $\pi_2$,
by Lemma \ref{lem:node_split}, 
$\pi_1$ and $\pi_2$ split at $v$ with probability at least $|d_1 - d_2|$ (over the random choice of $\eps_{i+1}$).
Moreover, since $\pi_1$ and $\pi_2$ are unique shortest paths in $G$ and $G$ is acyclic, it follows that $\pi_1 \cap \pi_2$ is a contiguous subpath of $\pi_1$ and $\pi_2$; thus, once they split, they can no longer intersect in later layers.
The number of nodes in the intersection $|\pi_1 \cap \pi_2|$ is $1$ more than the number of consecutive nodes at which $\pi_1, \pi_2$ intersect but do not split.
So by the above discussion, we have
\begin{align*}
    \Pr\left[|\pi_1 \cap \pi_2| > \frac{8 \log n}{ |d_1 - d_2|}\right] & \leq (1 - |d_1 - d_2|)^{\frac{8 \log n}{ |d_1 - d_2|}}\\
    & \leq e^{-|d_1 - d_2| \cdot \frac{8 \log n}{ |d_1 - d_2|}}\\
    & \leq e^{-8 \log n} \\
    & = n^{-8} 
\end{align*}

For the second point of the lemma: if $|d_1, d_2| \ge 1$, then by Lemma \ref{lem:node_split}, if there is a node $v \in \pi_1 \cap \pi_2$, then $\pi_1$ and $\pi_2$ split at $v$ with probability $1$, and then they can no longer intersect in later layers.
So we have $|\pi_1 \cap \pi_2| \leq 1$.
\end{proof}
Since $|\Pi| = p \leq n^2$, we can argue by a union bound that Lemma \ref{lem:path_overlap} holds for all $\pi_1, \pi_2 \in \Pi$ simultaneously with probability at least $1 - n^{-4}$.  From now on, we will assume  that this property holds for our constructed graph $G$.


Once we specify our construction parameters $\ell$ and $q$, the following lemma will immediately imply the third property of $G$ as stated in Lemma \ref{lem:construction_lemma}. 


\begin{lemma}
Let $x, y \in V$ be distinct nodes in $G$, and let $z =\hopdist_G(x, y)$.
Let
$$
\Pi' := \{ \pi \in \Pi \mid x, y \in \pi \}. 
$$
Then $|\Pi'| \leq  \max\left\{\frac{16q \log n}{z}, 1 \right\}$.
\label{lem:no_good_hops}
\end{lemma}
\begin{proof}

Let $\Pi' = \{\pi_1, \dots, \pi_k\}$ and let $(1, d_i) \in D$ be the  direction vector associated with $\pi_i$ for $i \in [1, k]$.
Since the paths in $\Pi'$ all intersect, by Proposition \ref{prop:dist_dir} we must have $d_i \neq d_j$ for $i \neq j$.  
Let $a = \min_{i \in [1, k]} d_i$ and let $b = \max_{i \in [1, k]} d_i$. Then 
$$
d_b - d_a \geq \frac{k - 1}{q},
$$
since $|d_i - d_j| \geq 1/q$ for all $(1, d_i), (1, d_j) \in D$ such that $i \neq j$.
Thus, by Lemma \ref{lem:path_overlap} we must have $\frac{k-1}{q} < 1$, since we have at least two nodes $x, y \in \pi_a \cap \pi_b$.
So by Lemma \ref{lem:path_overlap}, 
$$
|\pi_a \cap \pi_b| \leq \frac{8 \log n}{d_b - d_a} \leq \frac{8 q \log n}{k-1}.
$$

Since $x, y \in \pi_a \cap \pi_b$ and $\pi_a$ and $\pi_b$ are unique shortest paths in $G$, it follows that they coincide on their subpaths $\pi_a[x, y] = \pi_b[x, y]$.
Moreover, since the hopdistance from $x$ to $y$ in $G$ is $z$, it follows that $|\pi_a[x, y]| = |\pi_b[x, y]| = z + 1$.
Then taken together we have 
$$
z \leq |\pi_a \cap \pi_b| \leq  \frac{8 q \log n}{k-1}.
$$
Rearranging, we get
$$k \leq \frac{8q \log n}{z} + 1.$$
If $\frac{8q\log n}{z} \geq 1$, then this implies $k \leq \frac{16q \log n}{z}$.
Otherwise, if $\frac{8q\log n}{z} < 1$, then this implies that $k \leq 1$ since $k \in \mathbb{Z}$.  
%
%
%
%
\end{proof}

\subsection{Finishing the proof of Lemma \ref{lem:construction_lemma}}


We note that Theorem \ref{thm:eh} is trivial in the parameter regime $p = \Omega(n^2 / \log n)$, since its lower bound on hopbound is $\Omega(1)$.
So we will assume $p = O(n^2 / \log n)$ in the following, with as small of an implicit constant as needed.
Let
$$\ell =  \frac{n}{2^{10}p^{1/2} \log^{1/2} n} \text{\qquad and \qquad} q = \frac{\ell}{2^{10}\log n}.$$ 

We now quickly verify that graph $G$ and associated critical paths $\Pi$ satisfy the properties of Lemma \ref{lem:construction_lemma}:
\begin{itemize}
    \item By construction, $G$ has $\ell = \Theta\left( \frac{n}{p^{1/2}\log^{1/2}n} \right)$ layers, and each path in $\Pi$ travels from the first layer to the last layer.  
    \item  Proposition \ref{prop:dist_dir} implies that each vertex has at most $|D| = O\left( qn / \ell^2 \right)$ paths passing through it. By our choice of construction parameters $\ell$ and $q$, we conclude that $|D| = O(\ell \cdot p / n)$.
    \item Each path $\pi \in \Pi$ is a unique shortest path between its endpoints in $G$ by Lemma \ref{lem:unique_paths}. 
    \item Since $16q \log n \leq \ell$ and $\hopdist_G(u, v) \leq \ell$, Lemma \ref{lem:no_good_hops} immediately implies that  for all $u, v \in V$, there are at most 
    $$ \max\left\{\frac{16q \log n}{\hopdist_G(u, v)}, 1 \right\}  \leq \frac{\ell}{\hopdist_G(u, v)} $$
    paths in $\Pi$ that contain both $u$ and $v$. 
    \item For all critical pairs $(s, t) \in P$, the hopdistance from $s$ to $t$ in $G$ is $\ell - 1 \gg 16q \log n$. Then by Lemma \ref{lem:no_good_hops}, each of the $|\Pi|$ paths constructed in Section \ref{subsec:critical_paths} have distinct endpoints and thus are distinct. Then
    $$
    |\Pi| = |S||D| \geq \frac{n}{2\ell} \cdot \frac{nq}{8\ell^2} \geq \frac{n^2q}{16\ell^3} \geq p.
    $$ 
\end{itemize}

We have shown that our \textit{directed} graph $G$ satisfies the properties of Lemma \ref{lem:construction_lemma} in the regime of $p \in [n, n^2]$. Moreover,  our construction still goes through  even in the extended regime of $p \in [n/c, n^2]$ for any constant $c>0$. All that remains is to extend our construction to the entire regime of $p \in [1, n]$ and make $G$ undirected.

\paragraph{Extending the construction to $p \in [1, n]$.}
We can extend our construction to the regime of $p \in [1, n]$ with a simple modification to $G$ that was previously used in the prior work of \cite{KP22a}. We will sketch the modification here and defer the proof of correctness to Lemma  \ref{lem:small_p} in Appendix \ref{app:small_p}.

Let $G(n, p)$ denote an instance of our originally constructed graph $G$ with input parameters $n$ and $p \in [n, n^2]$. Let $n$ be a sufficiently large integer and let $p \in [1, n]$. Let $G_1 := G(p_1, p)$ where $p_1 = \Theta(p)$ and $p_1$ divides $n$.  Now for each node $v$ in $G_1$, replace $v$ with a directed path $\pi_v = (u_1^v, \dots, u_{k}^v)$ with $ k = n/p_1$ nodes. For all $v \in V$, assign weight $0$ to all edges in $\pi_v$. For each edge $(v_1, v_2)$ originally in $G_1$, add edge $(u_{k}^{v_1}, u_1^{v_2})$ to the graph. Let $G_2$ be the resulting graph, and let $\Pi_2$ be the updated set of critical paths. This completes the modification.
\begin{lemma}
 The $n$-node graph $G_2$ and the set $\Pi_2$ of  $|\Pi_2| \geq p$ paths satisfy the properties of Lemma \ref{lem:construction_lemma}.
\label{lem:small_p}
\end{lemma}
\begin{proof}
    We defer the proof to Appendix \ref{app:small_p}, as it largely follows our earlier analysis. 
\end{proof}

\paragraph{Making $G$ undirected.}

To make $G$ undirected, we use the following standard simple blackbox reduction.
Let $W$ be the sum of all edge weights in $G$, i.e., $W = \sum_{e \in E} w(e)$.
For every edge $e \in E$, add $+W$ to the weight $w(e)$ of $e$, and treat $e$ as an undirected edge.
Call the resulting graph $G'$.

We now argue correctness: in particular, we need to argue that for all $s, t \in V$ such that $t$ is reachable from $s$ in $G$, $\pi$ is a shortest weighted (directed) $s \leadsto t$-path in $G$ if and only if $\pi$ is a shortest weighted (undirected) $s \leadsto t$-path in $G'$. 
\begin{itemize}
\item First, note that for all $s \leadsto t$-paths $\pi'$ in $G'$, the number of nodes in $\pi'$ satisfies $|\pi'| \geq |\pi|$ by the construction of $G$ and $G'$.

\item Moreover, if $|\pi'| > |\pi|$, then $\pi'$ has one more edge than $\pi$.
Thus, its weighted length in $G'$ satisfies
\begin{align*}
w(\pi') &> (|\pi|+1)\cdot W\\
&= W|\pi| + \sum \limits_{e \in E(G)} w(e)\\
&> w(\pi)
\end{align*}
and so $\pi'$ is not a shortest path.
We conclude that if $\pi'$ is a shortest $s \leadsto t$-path in $G'$, then $|\pi'| = |\pi|$. 

\item Any $s \leadsto t$-path $\pi'$ in $G'$ with $|\pi|$ edges must use exactly one node in each layer, and thus it respects the original edge directions in $G$.
We conclude that $\pi'$ is a shortest weighted $s \leadsto t$-path in $G'$ if and only if $\pi'$ is a shortest weighted (directed) $s \leadsto t$-path in $G$. 
\end{itemize}

\noindent 
Lemma \ref{lem:construction_lemma}  is immediate from the above discussion.


%% file: shortcut.tex
\section{Shortcut Sets}

In this section we will prove the following theorem.

\begin{theorem} \label{thm:ss}
For any parameter $p \in [1,  n^{5/4}]$, there exists an $n$-node unweighted directed graph $G = (V, E)$ such that for any shortcut set $H$ of size $|H| \leq p$, the graph $G \cup H$ must have diameter $D$, where
\[ D = \begin{cases} 
      \Omega\left( \frac{n}{p^{3/4} \log n} \right) & \text{for } p \in [1, n / (\alpha \log n) ], \text{ where $\alpha > 0$ is a sufficiently large constant} \\
      \Omega\left( \frac{n^{5/4}}{p \log^{9/4} n} \right) & \text{for } p \in [n/(\alpha \log n), n^{5/4}]. \\
   \end{cases}
\]
In particular, when $p = O(n)$, $G \cup H$ must have diameter $D = \Omega\left(\frac{n^{1/4}}{\log^{9/4}n} \right)$. 
\end{theorem}

\noindent
We will prove this via a construction of the following type:

\begin{lemma}
    \label{lem:ss_construction_lemma}
    For any $p \in [1, n/ (\alpha \log n)]$, where $\alpha > 0$ is a sufficiently large constant, there is an infinite family of $n$-node directed unweighted graphs $G = (V, E, w)$ and sets $\Pi$ of $|\Pi|=p$ paths in $G$ with the following properties:
\begin{itemize}
\item $G$ has $\ell = \Theta\left( \frac{n}{p^{3/4} \log n } \right)$ layers. 
Each path in $\Pi$ starts in the first layer, ends in the last layer, and contains exactly one node in each layer. 

\item Each path in $\Pi$ is the unique  path between its endpoints in $G$.

\item For any two nodes $u, v \in V$, there are at most $ \frac{\ell}{\hopdist_G(u, v)} $ paths in $\Pi$ that contain both $u$ and $v$.
\end{itemize}
\end{lemma}

\noindent
We quickly verify that Lemma \ref{lem:ss_construction_lemma} implies Theorem \ref{thm:ss}.

\begin{proof}[Proof of Theorem \ref{thm:ss} using Lemma \ref{lem:ss_construction_lemma}]
Fix an $n$ and $p \in [1, n/(\alpha \log n)]$. Let $G = (V, E, w)$ be the graph in Lemma \ref{lem:ss_construction_lemma} with associated set $\Pi$ of $|\Pi| = 2p$ paths in $G$. Let $H$ be a shortcut set of size $|H| \leq p$. Let $P \subseteq V \times V$ be the set of node pairs that are the endpoints of paths in $\Pi$. Since all paths in $\Pi$ are \textit{unique} paths between their endpoints in $G$, it follows that 
$$
\texttt{diameter}(G \cup H) \geq \max_{(s, t) \in P} \dist_{G \cup H}(s, t) = \max_{(s, t) \in P} \hopdist_{G \cup H}(s, t).
$$
Then when $p \in [1, n/(\alpha \log n)]$, we  can achieve the bounds in Theorem \ref{thm:ss} using the same potential function argument as in Section \ref{subsec:eh_thm}.

In order to extend our shortcut set lower bound  to $p \in [n/(\alpha \log n), n^{5/4}]$, we will appeal to a more general property of the extremal functions of exact hopsets and shortcut sets, which we defer to Lemma \ref{lem:path_subsample} in Appendix \ref{app:path_subsample}. 
\end{proof}

\subsection{Constructing the strongly convex vector set $W(q)$ }


In our construction of the graph $G$, we will implicitly use the following lemma from \cite{BL98, BH22} that establishes the existence of a dense set of vectors that each extend the farthest in their own direction.\footnote{This is a slightly stronger property than convexity, and hence is sometimes referred to as ``strong convexity'' in the area.}

\begin{lemma}[Theorem 1 of \cite{BL98}; Lemma 7 of \cite{BH22}]
\label{lem:convex_set}
    For sufficiently large $r \in \mathbb{Z}_{\geq 0}$, there exists a strongly convex set of integer vectors $W(r) \subseteq \mathbb{Z}^2$ of size $|W(r)| = \Theta(r^{2/3})$, such that
    \begin{itemize}
        \item for all $\Vec{w} \in W(r)$, $\|\Vec{w}\| \leq r$, 
        \item every $\Vec{w} \in W(r)$ lies in the first quadrant, i.e., both coordinates are positive, and
        \item for all distinct $\Vec{w}_1, \Vec{w}_2 \in W(r)$,  $\langle \Vec{w}_1 , \Vec{w}_2\rangle < \langle \vec{w}_1, \Vec{w}_1 \rangle$.
    \end{itemize}
\end{lemma}

In our construction of $G$, we will use a vector set $W$ from this lemma to help generate edge and direction vectors.
We will make use of the following technical property of the vectors in $W$:

\begin{prop}
\label{prop:vec_proj_order}
   Let $W = \{ \Vec{w}_1, \dots, \Vec{w}_q \}$ be a set of vectors as described in Lemma \ref{lem:convex_set}, with its vectors ordered counterclockwise. For all $\Vec{w}_i, \Vec{w}_j, \Vec{w}_k \in W$ with $i < j < k$, the following inequalities hold:
   $$
   \left\langle \Vec{w}_i, \Vec{w}_k \right\rangle < \left\langle \Vec{w}_j, \Vec{w}_k \right\rangle < \left\langle \Vec{w}_k, \Vec{w}_k \right\rangle  \text{\qquad and \qquad} \left\langle \Vec{w}_i, \Vec{w}_k \right\rangle < \left\langle \Vec{w}_i, \Vec{w}_j \right\rangle < \left\langle \Vec{w}_i, \Vec{w}_i \right\rangle.
   $$
\end{prop}
\begin{proof}
We will only prove here that $\langle \Vec{w}_i, \Vec{w}_k \rangle < \langle \Vec{w}_j, \Vec{w}_k \rangle < \langle \Vec{w}_k, \Vec{w}_k \rangle$; the other set of inequalities follows from an identical argument. 
By Lemma \ref{lem:convex_set} we already have $\langle \Vec{w}_j, \Vec{w}_k \rangle < \langle \Vec{w}_k, \Vec{w}_k \rangle$, so it remains only to show that $\langle \Vec{w}_i, \Vec{w}_k \rangle < \langle \Vec{w}_j, \Vec{w}_k \rangle $.

Let $\psi_1$ be the inner angle between $\Vec{w}_i$ and $\vec{w}_j$ and let $\psi_2$ be the inner angle between $\Vec{w}_j$ and $\vec{w}_k$; thus the inner angle between $\Vec{w}_i$ and $\vec{w}_k$ is $\psi_1 + \psi_2$. 
We first establish a useful inequality:
\begin{align*}
    \langle \Vec{w}_i, \Vec{w}_j \rangle & < \langle \Vec{w}_j, \Vec{w}_j \rangle \tag*{ by Lemma \ref{lem:convex_set}} \\
    \|\Vec{w}_i\| \|\Vec{w}_j\| \cos \psi_1 & < \|\Vec{w}_j\|^2 \tag*{cosine formula} \\
        \|\Vec{w}_i\|\|\Vec{w}_k\|  \cos \psi_1  \cos \psi_2 & < \|\Vec{w}_j\| \|\Vec{w}_k\| \cos \psi_2. \tag*{follows algebraically from previous line}
\end{align*}
We will next use the trigonometric identity
\begin{align*}\cos \psi_1 \cos \psi_2 &= \sin \psi_1 \sin \psi_2 + \cos(\psi_1 + \psi_2)\\
&> \cos(\psi_1 + \psi_2) \tag*{since \text{$\psi_1, \psi_2 \in (0, \pi/2)$} by Lemma \ref{lem:convex_set}.}
\end{align*}
We are now ready to show:
\begin{align*}
\langle \Vec{w}_i, \Vec{w}_k \rangle &=  \|\Vec{w}_i\|\|\Vec{w}_k\|\cos(\psi_1 + \psi_2) \tag*{cosine formula}\\
&< \|\Vec{w}_i\|\|\Vec{w}_k\|  \cos \psi_1  \cos \psi_2 \tag*{second inequality}\\
&< \|\Vec{w}_j\| \|\Vec{w}_k\| \cos \psi_2 \tag*{first inequality}\\
&< \langle \Vec{w}_j, \Vec{w}_k \rangle. \tag*{cosine formula \qedhere}
\end{align*}
\end{proof}

\subsection{Constructing $G$}

We next construct the graph that will be used for Lemma \ref{lem:ss_construction_lemma}. Let $n$ be a sufficiently large positive integer, and let  $p = n / (\alpha_0 \log n)$ for a sufficiently large constant $\alpha_0 > 0$ to be chosen later (we will extend our construction to other choices of $p \in [1, n^{5/4}]$ later).
For simplicity of presentation, we will frequently ignore issues related to non-integrality of expressions that arise in our construction; these issues affect our bounds only by lower-order terms. 
 All edges $(u, v)$ in $G$ will be directed from $u$ to $v$. 

\paragraph{Vertex Set $V$.}
\begin{itemize}
    \item Let $r$ be a positive integer construction parameter to be specified later. Our graph $G$ will have $\ell := \frac{n^{1/3}}{r^{2/3}}$ layers $L_1, \dots, L_{\ell}$, and each layer will have $n/\ell = n^{2/3}r^{2/3}$ nodes.
    \item We will label each node in layer $L_i$, $i \in [1, \ell]$, with a distinct triple in $\{i\} \times \left[1,  n^{1/3}r^{1/3} \right] \times \left[1,  n^{1/3}r^{1/3}\right]$. We will interpret the node in $L_i$ labeled $(i, j, k)$ as an integer point $(i, j, k) \in \mathbb{Z}^3$.
    \item These $n$ nodes arranged in $\ell$ layers will compose our node set $$V = [1, \ell] \times [1,  n^{1/3}r^{1/3}] \times  [1,  n^{1/3}r^{1/3}] \subseteq \mathbb{Z}^3$$ of graph $G$. In a slight abuse of notation, we will treat $v \in V$  as either a node in $G$ or a point in $\mathbb{Z}^3$, depending on the context. 
\end{itemize}

\paragraph{Edge Set $E$.} 

\begin{itemize}
    \item All the edges in $E$ will be between consecutive layers $L_i, L_{i+1}$ of $G$. We will let $E_i$ denote the set of edges in $G$ between layers $L_i$ and $L_{i+1}$.
    \item Just as our nodes in $V$ correspond to points in $\mathbb{R}^3$, we can interpret the edges $E$ in $G$ as vectors in $\mathbb{R}^3$. In particular, for every edge $e = (v_1, v_2) \in E$, we identify $e$ with  the corresponding vector $\Vec{u}_e := v_2 - v_1$. Note that since all edges in $E$ are between adjacent layers $L_i$ and $L_{i+1}$, the first coordinate of $\Vec{u}_e$ is $1$ for all $e \in E$. We will use $u_{e}^i$ to denote the $i$th coordinate of $\Vec{u}_e$ for $i \in \{2, 3\}$ (i.e. for all $e \in E$, we write $\Vec{u}_e = (1, u_{e}^2, u_{e}^3)$).
    \item We begin our construction of $E$ by defining the set of vectors $W := W(r^3) \subseteq \mathbb{Z}^2$, where $W(r^3)$ is the strongly convex set of vectors defined in Lemma \ref{lem:convex_set}. Let $W = \{\Vec{w}_1, \dots, \Vec{w}_q\} $, where vectors $\Vec{w}_i$ are ordered counterclockwise and $q := |W| = \Theta(r^2)$. We may assume wlog that $q \leq r^2$ (e.g. by removing vectors from $W$ until this is true). 
    \item Using $W$, we define our set of \textit{edge vectors} $C \subseteq \mathbb{Z}^2$ as:
    $$
    C := \left\{ \Vec{w}_{i+1} - \Vec{w}_i \mid i \in [1, q-1] \right\}.
    $$
    Let $\Vec{c}_i := \Vec{w}_{i+1} - \Vec{w}_i $ denote the $i$th vector of $C$. 
    \item For each layer $L_i$, $i \in [1, \ell - 1]$, sample a random integer  from $[1, q-1]$ and call it $\lambda_i$.    
    We define the set $C_i \subseteq \mathbb{Z}^2$  as  
    $$
    C_i := \{(0, 0), \Vec{c}_{\lambda_i}\},
    $$
    where $\Vec{c}_{\lambda_i}$ is the $\lambda_i$th vector of $C$. 
    Note that $C_i$ contains exactly two vectors, the zero vector and a randomly chosen vector $\Vec{c}_{\lambda_i}$ from $C$. Intuitively: for each layer $L_i$, we are sampling two adjacent vectors $\Vec{w}_{\lambda_i}$ and $\vec{w}_{\lambda_i + 1}$ from $C$ and adding $-\vec{w}_{\lambda_i}$ to each of them to obtain $C_i$. The purpose of adding the normalizing vector $-\vec{w}_{\lambda_i}$ to $\Vec{w}_{\lambda_i}$ and $\vec{w}_{\lambda_i + 1}$ is to reduce the magnitude of the vectors in $C_i$, as we will formalize in Proposition \ref{prop:small_c}.
    
    \item The vectors in set $C_i$ will define the edges in $E_i$. Specifically, for all $(i, v_1, v_2) \in L_i$ and for all $(c_1, c_2) \in C_i$ such that 
    $$
    (i+1, v_1 + c_1, v_2 + c_2) \in L_{i+1},
    $$
    we add the  edge $((i, v_1, v_2), (i+1, v_1 + c_1, v_2 + c_2))$ to $E_i$.
\end{itemize}
This completes the construction of our graph $G = (V, E)$.  We now verify that the vectors in $C_i$ have small magnitude in expectation.

\begin{prop}
    For all $i \in [1, \ell - 1]$, $\mathbb{E}[\|\Vec{c}_{\lambda_i}\|] \leq \frac{2^4r^3}{q} = \Theta(r)$.
    \label{prop:small_c}
\end{prop}
\begin{proof}
    Note that the vectors in $C$ correspond to sides of a convex polygon whose vertices are the vectors in $W$. Since this polygon is contained in a ball of radius $r^3$ in $\mathbb{R}^2_{> 0}$ by Lemma \ref{lem:convex_set}, it follows that $\sum_{\Vec{c} \in C} \|\Vec{c}\| \leq 2\pi r^3$. Note that $|C| = q - 1\geq q/2$ for sufficiently large $q$. Then
    $$
        \mathbb{E}[\|\Vec{c}_{\lambda_i}\|] = \sum_{\Vec{c} \in C}  \frac{1}{|C|} \cdot \|\Vec{c}\|  \leq \frac{2\pi r^3}{|C|} \leq \frac{2^4r^3}{q}. \qedhere
    $$
\end{proof}

The vectors $\Vec{w}_i \in W$ have magnitude roughly $r^3$, whereas the vectors in $C_i$ have expected magnitude at most $\Theta(r)$ by Proposition \ref{prop:small_c}. Since each edge in our critical path corresponds to a vector in $C_i$, ensuring that the vectors in $C_i$ have small magnitude (at least in expectation) will be essential for guaranteeing that the paths $\pi \in \Pi$ are long. 

Let us comment here on a discrepancy between this construction and the technical overview.
In the technical overview, we stated that we would use the convex vectors $W$ to generate the edge vectors $C$.
Instead, we are using the \emph{difference} between adjacent convex vectors $W$ to generate $C$.
This is an optimization: our plan is to sample \emph{one} edge vector from $C$, and use it together with the zero vector $(0, 0)$ as the two available vectors between pairs of adjacent layers.
This is equivalent to sampling \emph{two adjacent} edge vectors from $W$, as advertised in the technical overview, and then applying an appropriate translation of the next layer in space.
Our strategy lets us use vectors of length $\Theta(r)$, instead of $\Theta(r^3)$, and these shorter vectors ultimately lead to a stronger lower bound.


\subsection{Direction Vectors, Critical Pairs, and Critical Paths}

\label{subsec:ss_critical_paths}

Our next step is to generate a set of \emph{critical pairs} $P \subseteq V \times V$, as well as a set of \emph{critical paths} $\Pi$.
Specifically, there will be one critical path $\pi_{s,t} \in \Pi$ going between each critical pair $(s, t) \in P$, and we will show that $\pi_{s,t}$ is the unique  $s \leadsto t$ path in $G$.

\paragraph{Direction Vectors $D$.}
We choose our set of direction vectors $D$ to be $D := W$, where $W$ is our strongly convex set of $q = \Theta(r^2)$ vectors. We will let $D = \{\vec{d}_1, \dots, \vec{d}_q \}$ be our list of direction vectors, and we will let the $i$th vector $\vec{d}_i$ of $D$ correspond to the $i$th vector $\vec{w}_i$ of $W$, i.e. $\vec{d}_i := \vec{w}_i$, for $i \in [1, q]$.
We will simply use the name $\vec{d}_i$ when we wish to emphasize the role of these vectors as direction vectors.

Note that since $D = W$, Proposition \ref{prop:vec_proj_order} also holds for $D$. That is, if $i, j, k \in [1, q-1]$ and $i < j < k$, then $$\langle \Vec{d}_i, \Vec{d}_k \rangle < \langle \Vec{d}_j, \Vec{d}_k \rangle < \langle \Vec{d}_k, \Vec{d}_k \rangle \qquad \text{ and } \qquad \langle \Vec{d}_i, \Vec{d}_k \rangle < \langle \Vec{d}_i, \Vec{d}_j \rangle < \langle \Vec{d}_i, \Vec{d}_i \rangle.$$




\paragraph{Critical Pairs $P$ and Critical Paths $\Pi$.}
\begin{itemize}
    \item We first define a set $S \subseteq L_1$, containing a subset  of the nodes in the first layer $L_1$ of $G$:
    $$
    S := \left\{  (1, j, k) \bigm\vert (j, k) \in \left[\frac{1}{3}n^{1/3}r^{1/3}, \hspace{2mm} \frac{2}{3}n^{1/3}r^{1/3} \right] \times \left[\frac{1}{3}n^{1/3}r^{1/3}, \hspace{2mm} \frac{2}{3}n^{1/3}r^{1/3}\right]   \right\}.
    $$
    Informally, $S$ is a middle square patch of the nodes in $L_1$. The key property of $S$ is that all nodes in $S$ are of distance at least $\frac{1}{3}n^{1/3}r^{2/3}$ from the sides of the square grid corresponding to layer $L_1$.   
    
    We will choose our set of demand pairs $P$ so that $P \subseteq S \times V$.  For every node $s \in S$ and direction vector $\vec{d} \in D$, we will choose a critical pair $(s, t) \in S \times V$ and a corresponding critical path $\pi_{s, t}$ to add to $P$ and $\Pi$.

\item 
Let $v_1 \in S$, and let $\Vec{d} = (1, d) \in D$.
The associated path $\pi$ has start node $v_1$.
We iteratively grow $\pi$, layer-by-layer, as follows.
Suppose that currently $\pi = (v_1, \dots, v_i)$ with each $v_i \in L_i$.
To determine the next node $v_{i+1} \in L_{i+1}$, let $E_i^{v_i} \subseteq E_i$ be the edges in $E_i$ incident to $v_i$, and let $\Vec{u}_i = (u_i^1, u_i^2) \in C_i$ be 
$$
\Vec{u}_i := \text{argmax}_{\Vec{c} \in C_i}\langle\Vec{c}, \Vec{d} \rangle.
$$

If $(1, u_i^1, u_i^2) = \Vec{u}_{e_i}$ for some $e_i \in E_i^{v_i}$, then by definition, $e_i$ is an edge whose first node is $v_i$; we define $v_{i+1} \in L_{i+1}$ to be the other node in $e_i$, and we append $v_{i+1}$ to $\pi$.   Otherwise, if there is no such edge $e_i$ in $E_i^{v_i}$, then we terminate our construction of path $\pi$ (i.e. $v_i$ will be the final node in $\pi$).



 This completes our construction of $P$ and $\Pi$. We will show that the paths generated in this way have distinct endpoints (with high probability), and therefore $|P| = |S||D| \geq \frac{n^{2/3}r^{2/3}q}{2^5}$, where $q = \Theta(r^{2})$.



\end{itemize}
 
 An important feature for correctness of our construction is that, when we iteratively generate paths, if we reach a point  $v_{i}$ such that $v_i + C_i \not \subseteq L_{i+1}$ (i.e. $v_i + \Vec{c} \not \in L_{i+1}$ for some $\Vec{c} \in C_i$), then we end our path at $v_i$. As a consequence, our critical paths in $\Pi$ may not travel through all $\ell$ layers of $G$. However, with nonzero probability, paths in $\Pi$  travel through a constant fraction of layers, as we prove in the following proposition.

\begin{prop}
  Let $\hat{\ell} := \ell \cdot \frac{q}{2^7r^2} = \Theta(\ell)$.  With probability at least $1/2$, for all $\pi \in \Pi$, $|\pi| \geq \hat{\ell}$. 
  \label{prop:path_lengths_ss}
\end{prop}
\begin{proof}
    Each critical path starts at a node $s$ in $S$ Each edge of the path  corresponds to a vector $(1, c_1, c_2) \in \mathbb{Z}^3$ such that $(c_1, c_2) \in C_i$ for some $i \in [1, \ell - 1]$. Our path ends when we reach the boundary  of our vertex set $V = [1, \ell] \times [1, n^{1/3}r^{1/3}] \times [1, n^{1/3}r^{1/3}] \subseteq \mathbb{Z}^3$. 
    We must show that before any such path travels through $\hat{\ell}$ nodes before reaching the boundary. Note that $\hat{\ell} = \ell \cdot \frac{q}{2^7r^2} \leq \ell$, since $q \leq r^2$.

    Let $x$ be the random variable defined as $x := \sum_{i = 1}^{\hat{\ell}} \Vec{c}_{\lambda_i}$. Then by Proposition \ref{prop:small_c} and the linearity of expectation, 
    $$\mathbb{E}[x] \leq \hat{\ell} \cdot \frac{2^4r^3}{q} = \ell \cdot \frac{r}{2^3} = \frac{1}{2^3}n^{1/3}r^{1/3},$$
    where the final equality follows from the fact that $\ell = n^{1/3}r^{-2/3}$. 
    Then by Markov's inequality, $$\Pr\left[x \leq  \frac{1}{4} n^{1/3}r^{1/3} \right] \geq 1/2.$$

    Now we claim that if $x \leq \frac{1}{4} n^{1/3}r^{1/3}$, then for all $\pi \in \Pi$, $|\pi| \geq  \hat{\ell}$. Let $\pi_{s, t} \in \Pi$ be a critical path for critical pair $(s, t) \in P$. Let $s = (s_1, s_2, s_3) \in \mathbb{Z}^3$ and let $t = (t_1, t_2, t_3) \in \mathbb{Z}^3$. By our construction of critical paths $\Pi$, either $t_1 = \ell$ or $t + C_{t_1} \not \in V$. In the first case, $|\pi| = \ell \geq \hat{\ell}$, since path $\pi$ traveled through all $\ell$ layers $L_i$. In the second case, we must have that $\|(t_2, t_3) - (s_2, s_3)\| \geq \frac{1}{3}n^{1/3}r^{1/3}$ by our choice of set $S$. But since $x \leq \frac{1}{4} n^{1/3}r^{1/3}$, we conclude that $|\pi| \geq \hat{\ell}$. The claim follows. 
\end{proof}

We have shown that with nonzero probability, all our paths in $\Pi$ travel through the first $\hat{\ell}$ layers of $G$. However, we cannot guarantee that paths in $\Pi$ travel to layers $L_i$ with $i > \hat{\ell}$.  Because of this, we choose to remove all layers $L_i$, $i > \hat{\ell}$, from $G$. We replace all critical paths $\pi \in \Pi$ with the truncated subpath of $\pi$ containing only the first $\hat{\ell}$ nodes of $\pi$, and we update our critical pairs $P \subseteq V \times V$ to be the set of all pairs of endpoints of the updated paths in $\Pi$.

\subsection{Critical paths are unique paths}

We now verify that graph $G$ and paths $\Pi$ have the unique path property as stated in Lemma \ref{lem:ss_construction_lemma}. This will follow straightforwardly from the properties of our set of direction vectors $D$, particularly Proposition \ref{prop:vec_proj_order}.

\begin{lemma}[Unique  paths]
\label{lem:ss_unique_paths}
    For every $(s, t) \in P$, path $\pi_{s, t}$ is a unique $s \leadsto t$-path in $G$. 
\end{lemma}
\begin{proof}
Fix a  direction vector $\Vec{d}_j$ in $D$. We claim that for all $i \in [1, \ell - 1]$, there is a unique vector $\Vec{c} \in C_i$ such that maximizes $\langle \Vec{c}, \Vec{d}_j \rangle$. Recall that $C_i \subseteq \mathbb{Z}^2$ contains exactly two vectors: the zero vector $(0, 0)$ and the vector $\Vec{c}_{\lambda_i} = \Vec{w}_{\lambda_i + 1} - \Vec{w}_{\lambda_i}$. 

Now assume  that $j \geq \lambda_i + 1$ and observe the following sequence of equivalent inequalities:
\begin{align*}
    \langle  \Vec{w}_{\lambda_i }, \Vec{w}_j  \rangle & < \langle \Vec{w}_{\lambda_i + 1}, \Vec{w}_j \rangle & \text{\qquad by Lemma \ref{lem:convex_set} and Proposition \ref{prop:vec_proj_order}}\\
    \langle  \Vec{0}, \Vec{w}_j  \rangle & < \langle \Vec{w}_{\lambda_i + 1} - \vec{w}_{\lambda_i}, \Vec{w}_j \rangle \\ 
    \langle  \Vec{0}, \Vec{d}_j  \rangle & < \langle \Vec{c}_{\lambda_i + 1}, \Vec{d}_j \rangle.
\end{align*}
When $j < \lambda_i + 1$, an identical argument shows that $\langle  \Vec{0}, \Vec{d}_j  \rangle  > \langle \Vec{c}_{\lambda_i + 1}, \Vec{d}_j \rangle$. In either case, there is a unique vector $\Vec{c} \in C_i$ maximizing $\langle \Vec{c}, \Vec{d}_j \rangle$ for all $i \in [1, \ell - 1]$.

Now fix a critical pair $(s, t) \in P$ that has $\Vec{d}_j \in D$ as its associated direction vector and $\pi_{s, t} = (v_1, \dots, v_k )$ as its critical path (where $s = v_1$ and $t = v_k$).  Let $f: \mathbb{R}^3 \mapsto \mathbb{R}^2$ be the function that projects each point in $\mathbb{R}^3$ onto the subspace corresponding to the last two coordinates of $\mathbb{R}^3$, i.e. $f(x, y, z) = (y, z)$ for all $(x, y, z) \in \mathbb{R}^3$.  

Let $\pi$ be an arbitrary $s \leadsto t$-path, and note that $|\pi| = |\pi_{s, t}|$ since $G$ is a layered directed graph. Let $\pi =  (v_1', \dots, v_k')$, where $s = v_1'$ and $t = v_k'$. By our construction of path $\pi_{s, t}$, we must have that  for all $i \in [1, k - 1]$, $$\langle f(v'_{i+1} - v'_i), \Vec{d}_j   \rangle \leq \langle f(v_{i+1} - v_i), \Vec{d}_j   \rangle.$$
Now suppose for the sake of contradiction that $\pi \neq \pi_{s, t}$. Then  $v'_{i+1} - v'_i \neq v_{i+1} - v_i$ for some $i \in [1, k - 1]$. Then $\langle f(v'_{i+1} - v'_i), \Vec{d}_j \hspace{1mm} \rangle < \langle f(v_{i+1} - v_i), \Vec{d}_j \hspace{1mm} \rangle$ by the above discussion. But then since $\pi$ and $\pi_{s, t}$ are both $s \leadsto t$-paths,
\begin{align*}
    \langle f(t - s), \Vec{d}_j \rangle = \sum_{i = 1}^{k-1 } \langle f(v'_{i+1} - v'_i), \Vec{d}_j  \rangle < \sum_{i = 1}^{k-1 } \langle f(v_{i+1} - v_i), \Vec{d}_j \rangle = \langle f(t - s), \Vec{d}_j  \rangle.
\end{align*}
This is a contradiction, so we conclude that $ \pi = \pi_{s, t}$. Then path $\pi_{s, t}$ is a unique $s \leadsto t$-path in $G$.
\end{proof}

\subsection{Critical Paths Intersection Properties}

Before finishing our proof of Lemma \ref{lem:ss_construction_lemma}, we will need to establish several properties of the critical paths in $\Pi$. 

\begin{prop}
Let $\pi_1, \pi_2 \in \Pi$ be two critical paths with the same corresponding direction vector $\vec{d} \in D$. Then $\pi_1 \cap \pi_2 = \emptyset$.
\label{prop:dist_paths_ss}
\end{prop}
\begin{proof}
    Let $k = \min\{ |\pi_1|, |\pi_2| \}$. 
    Let $v_i^j \in L_i$ denote the $i$th node of $\pi_j$, where $j \in \{1, 2\}$ and $i \in [1, k]$.  
    Note that since $\pi_1$ and $\pi_2$ share the same direction vector $\Vec{d}$, 
    edges $(v_i^1, v_{i+1}^1)$ and $(v_i^2, v_{i+1}^2)$ have the same corresponding vector $\Vec{u}_i \in C_i$ for all $i \in [1, k - 1]$  by our construction of $\pi_1$ and $\pi_2$. Also,  for each node $v \in L_1$,  $v$ belongs to at most one path $\pi \in \Pi$ with direction vector $\Vec{d}$, so $v_1^1 \neq v_1^2$. 
    Then for all $i \in [1, k]$,
    $$
    v_i^1 = v_1^1 + \sum_{i=1}^{i-1} \Vec{u}_i \neq v_1^2 + \sum_{i=1}^{i-1} \Vec{u}_i = v_i^2. \qedhere
    $$    
\end{proof}

Let $\pi_1, \pi_2 \in \Pi$ be two critical paths, and let $v \in V$ be a node in $G$. We say that paths $\pi_1$ and $\pi_2$ \textit{split} at $v$  if $v \in \pi_1 \cap \pi_2$ and 
the node following $v$ in $\pi_1$ is distinct from the node following $v$  in $\pi_2$, and we simply say that $\pi_1$ and $\pi_2$ \textit{split} if there exists some $v \in V$ such that $\pi_1$ and $\pi_2$ split at $v$. Note that since $\pi_1, \pi_2 \in \Pi$ are unique  paths in $G$, paths $\pi_1$ and $\pi_2$ can split at most once. 

\begin{lemma}
Fix a node $v \in L_i$, where $i \in [1, \ell - 1]$, and let $\pi_1, \pi_2 \in \Pi$ be critical paths with direction vectors $\vec{d}_j$ and $\vec{d}_k \in D$, $j, k \in [1, q]$, such that $v \in \pi_1$ and $v \in \pi_2$. Then paths $\pi_1$ and $\pi_2$  split at $v$ with probability at least $\frac{|j - k|}{ q }$. 
\label{lem:node_split_ss}
\end{lemma}
\begin{proof}
Fix a node $v \in L_i$, where $i \in [1, \ell - 1]$, and let $\pi_1, \pi_2 \in \Pi$ be critical paths with direction vectors $\vec{d}_j$ and $\vec{d}_k \in D$, $j, k \in [1, q]$, such that $v \in \pi_1$ and $v \in \pi_2$. By Proposition \ref{prop:dist_paths_ss}, $j \neq k$, and assume wlog that $j < k$.  Let $F$ be the event that the random variable $\lambda_i$ was sampled so that 
$$
j \leq \lambda_i < k.
$$
Our proof strategy is to show that $F$ implies that $\pi_1, \pi_2$ split at $v$, and then to show that $F$ occurs with the claimed probability.


\paragraph{$F$ implies that $\pi_1, \pi_2$ split at $v$.}
Assume that $F$ occurs. Then $j \leq \lambda_i < \lambda_i + 1 \leq k$. Now observe the following sequence of equivalent inequalities:
\begin{align*}
    \langle \Vec{w}_j, \Vec{w}_{\lambda_i} \rangle & > \langle \Vec{w}_j, \Vec{w}_{\lambda_i + 1} \rangle & \qquad \text{ by Lemma \ref{lem:convex_set} and Proposition \ref{prop:vec_proj_order}} \\
     \langle \Vec{w}_j, \Vec{0} \rangle & > \langle \Vec{w}_j, \Vec{w}_{\lambda_i + 1} - \Vec{w}_{\lambda_i}  \rangle \\
    \langle \Vec{d}_j, \Vec{0} \rangle & > \langle \Vec{d}_j, \vec{c}_{\lambda_i}  \rangle. 
\end{align*}
Since $C_i = \{\Vec{0}, \Vec{c}_{\lambda_i}\}$, by our construction of the critical paths in $\Pi$,  the above inequality $\langle \Vec{d}_j, \Vec{0} \rangle  > \langle \Vec{d}_j, \vec{c}_{\lambda_i}  \rangle$ implies that path $\pi_1$ takes an edge in $E_i$ corresponding to vector $\Vec{0}$. An identical argument will show that $\langle \Vec{d}_k, \Vec{0} \rangle < \langle \Vec{d}_k, \vec{c}_{\lambda_i}  \rangle$, so path $\pi_2$ takes an edge in $E_i$ corresponding to vector $\Vec{c}_{\lambda_i}$. Since paths $\pi_1$ and $\pi_2$ take different edges in $E_i$, they must split at $v$.

\paragraph{$F$ happens with good probability.} Random variable $\lambda_i$ is sampled uniformly from $[1, q - 1]$.  Then the event $F$ occurs with probability $\frac{|j - k|}{q-1} \geq \frac{|j - k|}{q}$.
\end{proof}

We will use Lemma \ref{lem:path_overlap_ss} to prove the following two lemmas, which capture key properties of our graph $G$. 

\begin{lemma}
    Let $\pi_1, \pi_2 \in \Pi$ be critical paths with associated direction vectors $\vec{d}_j, \vec{d}_k \in D$. Then $|\pi_1 \cap \pi_2| \leq \frac{8q \log n}{|j - k|}$ with probability at least $1 - n^{-8}$.
    \label{lem:path_overlap_ss}
\end{lemma}
\begin{proof}
If $\pi_1 \cap \pi_2 = \emptyset$, then the claim is immediate, so assume there is a node $v \in \pi_1 \cap \pi_2$. Suppose $v \in L_i$. 
    By Lemma \ref{lem:node_split_ss}, $\pi_1$ and $\pi_2$ split at $v$ with probability at least $\frac{|j-k|}{q}$. Moreover, conditioning on $v \in \pi_1 \cap \pi_2$, the event that $\pi_1$ and $\pi_2$ split at $v$ given that $v \in \pi_1$ and $v \in \pi_2$ depends only on our choice of $\lambda_i$ and is independent of $\lambda_j$ for $j \neq i$. 

Since $\pi_1$ and $\pi_2$ are unique  paths in $G$, it follows that $\pi_1 \cap \pi_2$ is a contiguous subpath of $\pi_1$ and $\pi_2$. The number of nodes in the intersection $|\pi_1 \cap \pi_2|$ is $1$ more than the number of consecutive nodes at which $\pi_1, \pi_2$ intersect but do not split.
So by the above discussion, we have
\begin{align*}
    \Pr\left[|\pi_1 \cap \pi_2| > \frac{8q \log n}{ |j - k| }\right] & \leq \left(1 -\frac{|j - k|}{q}\right)^{\frac{8q \log n}{|j - k| }}\\
    & \leq e^{-\frac{|j - k|}{q} \cdot \frac{8q \log n}{|j - k|}}\\
    & \leq e^{-8 \log n} \\
    & = n^{-8}. 
\end{align*}
\end{proof}

Since $|\Pi| = p \leq n^2$, we can argue by a union bound that Proposition \ref{prop:path_lengths_ss} holds and Lemma \ref{lem:path_overlap} holds for all $\pi_1, \pi_2 \in \Pi$ simultaneously with probability at least $1 - 1/2 - n^{-4} > 0$. From now on, we will assume that this property holds for our constructed graph $G$. 

Once we specify our construction parameters $\ell$ and $q$, the following lemma will immediately imply the third property of $G$ as stated in Lemma \ref{lem:ss_construction_lemma}.

\begin{lemma}
    Let $x, y \in V$ be nodes in $G$ such that the unweighted distance from $x$ to $y$ in $G$ is $z$, where $0 < z < \infty$. Let $\Pi' \subseteq \Pi$ be the following set of critical paths:
    $$
    \Pi' := \{  \pi \in \Pi \mid x, y \in \pi \}.
    $$
    Then $|\Pi'| \leq \max\left\{ \frac{16q \log n}{z}, 1 \right\}$.
    \label{lem:ss_no_good_hops}
\end{lemma}
\begin{proof}
    
Let $\Pi' = \{\pi_1, \dots, \pi_k\}$ and let $\Vec{d}_{\sigma_i} \in D$ be the  direction vector associated with $\pi_i$ for $i \in [1, k]$. By Proposition \ref{prop:dist_dir}, $\sigma_i \neq \sigma_j$ for $i \neq j$.  
Let $a = \min_{i \in [1, k]} \sigma_i$ and let $b = \max_{i \in [1, k]} \sigma_i$. Then $b-a \geq k-1$, so by Lemma \ref{lem:path_overlap_ss},
$$
|\pi_a \cap \pi_b| \leq \frac{8 q\log n}{ |b - a| } \leq \frac{8q \log n}{k-1}.
$$
Additionally, since $x, y \in \pi_a \cap \pi_b$ and $\pi_a$ and $\pi_b$ are unique  paths, it follows that $\pi_a[x, y] = \pi_b[x, y]$. Moreover, since the unweighted distance from $x$ to $y$ in $G$ is $z$, it follows that $|\pi_a \cap \pi_b| \geq |\pi_a[x, y]| = |\pi_b[x, y]| = z + 1$. Then taken together we have
$$
z \leq |\pi_a \cap \pi_b| \leq  \frac{8q \log n}{k-1}.
$$
Rearranging, we get that
$$k \leq \frac{8q \log n}{z} + 1 \leq   \max\left\{\frac{16q \log n}{z}, 1 \right\}.$$
(We obtain the final inequality from the following observation: if $\frac{8q\log n}{z} \geq 1$, then this implies $k \leq \frac{16q \log n}{z}$. Otherwise, if  $\frac{8q\log n}{z} < 1$, then  this implies that $k \leq 1$ since $k \in \mathbb{Z}$.)  
\end{proof}

\subsection{Finishing the proof of Lemma \ref{lem:ss_construction_lemma}}

Let $$r = \frac{n^{1/8}}{2^6 \log^{3/8}n}.$$ Now recall that $q = \Theta(r^2)$, and let $\alpha_1 > 0$ be a constant such that $\alpha_1r^2 \leq q \leq 2\alpha_1 r^2$. Then
$$
\ell = \frac{n^{1/3}}{r^{2/3}} \geq n^{1/4}\log^{1/4}n,
\qquad
\hat{\ell} =  \frac{q}{2^7r^2} \cdot \ell \geq \frac{\alpha_1 }{2^7} \cdot \ell,
\qquad
\text{and}
\qquad 
q \leq 2\alpha_1r^2 \leq \frac{\hat{\ell}}{16 \log n}.
$$
We now quickly verify that graph $G$ and associated critical paths $\Pi$ satisfy the properties of Lemma \ref{lem:ss_construction_lemma}:
\begin{itemize}
    \item By construction, $G$ has  $\hat{\ell} = \Theta\left(  n^{1/4}\log^{1/4}n \right)$ layers, and each path in $\Pi$ travels from the first layer to the $\hat{\ell}$th layer.  
    \item Each path $\pi \in \Pi$ is a unique  path between its endpoints in $G$ by Lemma \ref{lem:ss_unique_paths}. 
    \item Since $16q \log n \leq \ell$ and $\hopdist_G(u, v) \leq \ell$, Lemma \ref{lem:ss_no_good_hops} immediately implies that  for all $u, v \in V$, there are at most 
    $$ \max\left\{\frac{16q \log n}{\hopdist_G(u, v)}, 1 \right\}  \leq \frac{\ell}{\hopdist_G(u, v)} $$
    paths in $\Pi$ that contain both $u$ and $v$. 
    \item For all critical pairs $(s, t) \in P$, the unweighted distance from $s$ to $t$ in $G$ is $\ell - 1 \gg 16q \log n$. Then by Lemma \ref{lem:ss_no_good_hops}, each of the $|\Pi|$ paths constructed in Section \ref{subsec:ss_critical_paths} have distinct endpoints and thus are distinct. Then
    $$
    |\Pi| =|S||D|  \geq   \frac{n^{2/3}r^{2/3}q}{2^5} \geq \frac{\alpha_1 n^{2/3}r^{8/3}}{2^{5}}  \geq \frac{\alpha_1 n}{ 2^{21} \log n}.
    $$
    Recall that we initially let $p = n / (\alpha_0  \log n)$, for some unspecified constant $\alpha_0 > 0$. Choose $\alpha_0$ so that $p \leq |\Pi|$. 
\end{itemize}

We have shown that our directed graph $G$ satisfies the properties of Lemma \ref{lem:ss_construction_lemma} when $p = n / (\alpha_0 \log n)$. All that remains is to extend our construction to the regime $p \in [1, n / (\alpha_0 \log n)]$.



    

    


\paragraph{Extending the construction to $p \in [1, n / (\alpha_0 \log n)]$.}
We can extend our construction to the regime of $p \in [1, n / (\alpha_0 \log n)]$ using the same modification to $G$ that we performed on our exact hopset construction and that was previously used in the prior work of \cite{KP22a}. 
We will sketch the modification here. The proof of correctness follows from an argument identical to the proof of Lemma \ref{lem:small_p} in Appendix \ref{app:small_p}. 

We use $G(n, p)$ denote an instance of our originally constructed graph $G$ with input parameters $n$ and $p = n / (\alpha_0 \log n)$.  Let $n$ be a sufficiently large integer and let $p \in [1, n]$. Let $G_1 := G(p_1, p_1 / (\alpha_0 \log p_1))$ where $p = \Theta(p_1 / \log p_1)$ and $p_1$ divides $n$.  Now for each node $v$ in $G_1$, replace $v$ with a directed path $\pi_v = (u_1^v, \dots, u_{k}^v)$ with $ k = n/p_1$ nodes. For all $v \in V$, assign weight $0$ to all edges in $\pi_v$. For each edge $(v_1, v_2)$ originally in $G_1$, add edge $(u_{k}^{v_1}, u_1^{v_2})$ to the graph. Let $G_2$ be the resulting graph, and let $\Pi_2$ be the updated set of critical paths. This completes the modification.

%% file: appendix.tex
\section{Proof of Lemma \ref{lem:small_p}}
\label{app:small_p}
Let $n$ be a sufficiently large integer, and let $p \in [1, n]$ such that $p$ divides $n$. Let $p_1 := 2^{30}p$.  Recall that to obtain $G_2$, we first construct graph $G_1 := G(p_1, p)$, where $G(p_1, p)$ denotes our initial construction of the graph $G_1$ in Lemma \ref{lem:construction_lemma} on $p_1$ nodes with an associated set $\Pi_1$ of $|\Pi_1| = p$ paths. We let $V_1$ denote the nodes and $E_1$ denote the edges of $G_1$.  Let $P_1$ be the set of critical pairs associated with $G_1$,  and let $\Pi_1$ be the corresponding canonical paths. 
Let $\ell$ and $q$ be the construction parameters used to construct $G(p_1, p_1)$. Then
$$
2^{6}p \leq |P_1| \leq 2^{10}p, \text{ \qquad } \ell = \frac{p_1}{2^{10}p^{1/2} \log^{1/2}p_1}, \text{ \qquad and \qquad } q = \frac{p_1}{2^{20}p^{1/2} \log^{3/2}p_1}. 
$$
Note that for all $\pi \in \Pi_1$, we have that $|\pi| = \ell$. 
We then modified $G_1$ by replacing each node $v \in V_1$  with a path $\pi_v = (u_1^v, \dots, u_k^v)$ with $k = n/p_1$ nodes. If an edge $(v_1, v_2) \in E_1$ was originally in $G_1$,  we replaced it with an edge $(u_k^{v_1}, u_1^{v_2})$. This gave us our final $n$-node  graph $G_2 = (V_2, E_2)$.

Let $P_2 := \{(u_1^s, u_{k}^t) \mid (s, t) \in P_1\}$, and for all $(u_1^s, u_{k}^t) \in P_2$, let $\pi_{u_1^s, u_{k}^t}$ be the path obtained by taking $\pi_{s, t} \in \Pi_1$ and replacing each node $v \in \pi_{s, t}$ with the path $\pi_v$. Let $\Pi_2 := \{\pi_{s, t} \mid (s, t) \in P_2\}$. 
Then it is clear from our construction of $G_2$ that for all $(s, t) \in P_2$, the path  $\pi_{s, t} \in \Pi_2$ is a unique shortest $s\leadsto t$-path in $G_2$.  Additionally, for all $(s, t) \in P_2$, the number of nodes in path $\pi_{s, t} \in \Pi_2$ is at least $|\pi_{s, t}| \geq \ell_2$, where $$\ell_2 := \ell \cdot \frac{n}{p_1} = \frac{n}{2^{10}p^{1/2}\log^{1/2}p_1} = \Theta\left( \frac{n}{p^{1/2} \log^{1/2} p} \right).$$

\noindent
We now quickly verify that graph $G_2$ and associated critical paths $\Pi_2$ satisfy the properties of Lemma \ref{lem:construction_lemma}:
\begin{itemize}
    \item By construction, $G_2$ has $\ell_2 = \Theta\left(\frac{n}{p^{1/2} \log^{1/2} n}\right)$ layers, and each path in $\Pi$ travels from the first layer to the last layer. 
    \item Each path $\pi \in \Pi$ is a unique shortest path between its endpoints in $G_2$. This follows from  Lemma \ref{lem:unique_paths} and the observation that our path replacement step cannot increase the number of paths between pairs of nodes in $G_2$. 
    \item What remains is to show that for any two nodes $u, v \in V_2$, there are at most $\frac{\ell_2}{\hopdist_{G_2}(u, v)}$ paths in $\Pi_2$ that contain both $u$ and $v$. We prove this in two cases:
    \begin{itemize}
        \item \textbf{Case 1:} $u, v \in \pi_{w}$ for some $w \in V_1$. Note that $\hopdist_G(u, v) \leq n/p_1$. By Proposition \ref{prop:dist_dir}, the number of paths $\pi \in \Pi_1$ such that $w \in \pi$ is at most the number of direction vectors 
        \begin{align*}
            |D| & \leq \frac{p_1 q}{4\ell^2} \leq 
        \frac{p^{1/2}}{4 \log^{1/2} p_1}
        \leq
        \frac{p^{1/2}}{4 \log^{1/2} p_1} \cdot 
        \frac{n/p_1}{\hopdist_{G_2}(u, v)} \\
        &
        \leq
        \frac{n}{2^{30} p^{1/2} \log^{1/2}p} \cdot \frac{1}{\hopdist_{G_2}(u, v)} 
        \leq
        \frac{\ell_2}{\hopdist_{G_2}(u, v)}.
        \end{align*}
        Then by our construction of $G_2$, the number of paths $\pi \in \Pi_2$ such that $u, v \in \pi$ is at most $|D| \leq \frac{\ell_2}{\hopdist_{G_2}(u, v)}$. 
        
    \item \textbf{Case 2:} $u \in \pi_{w_1}$ and $v \in \pi_{w_2}$ for distinct $w_1, w_2 \in V_1$. By Lemma \ref{lem:no_good_hops}, the number of paths $\pi \in \Pi_1$ such that $w_1, w_2 \in \pi$ is at most
    \begin{align*}
    \max\left\{\frac{16 q \log n}{\hopdist_{G_1}(w_1, w_2)}, 1 \right\} & \leq \max \left\{
    \frac{\ell}{\hopdist_{G_1}(w_1, w_2)}, 1 
    \right\}
    \leq 
    \max \left\{
    \frac{\ell_2}{\hopdist_{G_2}(u, v)}, 1 
    \right\} \\
    & \leq \frac{\ell_2}{\hopdist_{G_2}(u, v)}.
    \end{align*}
    (The final inequality follows from the fact that $\hopdist_{G_2}(u, v) \leq \ell_2$.)
    \end{itemize}
\end{itemize}

\section{  Extending our shortcut set lower bound     }
\label{app:path_subsample}

To extend our shortcut set lower bound so that it holds in the regime of $p > n/(\alpha \log n)$, we will prove a more general statement about the behavior of the extremal function of shortcut sets. We write $\texttt{S}(n, p)$ for the smallest integer $D$ such that every $n$-node graph $G$ has a shortcut set $H$ of size $|H| \leq p$ such that $G \cup H$ has diameter at most $D$. 



\begin{lemma}
For all positive integers $n$ and $x \in [1, n]$,$$
\texttt{S}\left(n, \hspace{1mm} p/x \right) \leq  x \cdot \texttt{S}(n, \hspace{1mm} p)   \cdot 16\log n.
$$
\label{lem:path_subsample}
\end{lemma}
Lemma \ref{lem:path_subsample} essentially states that if we decrease the number of shortcuts allowed in our shortcut set, the extremal function controlling the worst-case size of shortcut sets won't increase by too much. We will use this lemma in the opposite direction to argue that our lower bound of $\texttt{S}(n, n/(\alpha \log n)) = \Omega\left(\frac{n^{1/4}}{\log^{1/4}n}\right)$ (where $\alpha > 0$ is a sufficiently large constant) that we obtained from our shortcut set construction in Lemma \ref{lem:ss_construction_lemma} implies lower bounds for shortcut sets with greater than $n/(\alpha \log n)$ shortcuts.

Let $p \in \left[\frac{n}{\alpha \log n}, n^{5/4}\right]$, and let $x = \frac{p}{n/(\alpha \log n)}$. Then by applying our lower bound from Lemma \ref{lem:ss_construction_lemma} to Lemma \ref{lem:path_subsample} we find that
$$
\Omega\left(\frac{n^{1/4}}{\log^{1/4}n}\right) \leq \texttt{S}\left(n, \frac{n}{\alpha \log n}\right) =
\texttt{S}\left(n, \frac{p}{x}\right) \leq x \cdot \texttt{S}(n, p) \cdot 16 \log n.
$$
Rearranging, we find that
$$
\texttt{S}(n, p) =  \Omega\left( \frac{n^{5/4}}{  p \log^{9/4} n } \right),
$$
as claimed in Theorem \ref{thm:ss}. 
\\

\noindent
We now prove Lemma \ref{lem:path_subsample}, which will follow from a simple path subsampling argument.

\begin{proof}[Proof of Lemma \ref{lem:path_subsample}]
Let $n$ be a positive integer, $p \in [1, n^2]$, and $x \in [1, n]$. Let $G = (V, E)$ be a graph on $n$ nodes. We subsample nodes of $G$ to construct a smaller graph $G_1 = (V_1, E_1)$ as follows.
\begin{itemize}
    \item Independently sample each node $V$ into set $V_1$ with probability $\frac{1}{2x}$. Then $\mathbb{E}[|V_1|] = \frac{n}{2x}$. 
    \item For all pairs of nodes $u, v \in V$ such that $\dist_G(u, v) \leq  8x \log n$, add directed edge $(u, v)$ to $E_1$. This completes the construction of $G_1 = (V_1, E_1)$.
\end{itemize}
By Markov's inequality, $|V_1| \leq n/x$ with probability at least $\frac{1}{2}$.
Assume for now that this does indeed hold, and we have $|V_1| \leq n/x$.  Then using  the fact that $\texttt{S}(c_1n, c_1p) \leq \texttt{S}(c_2n, c_2p)$ if $c_1 \leq c_2$, we find that there exists a shortcut set $H_1$ of size $|H_1| \leq p/x$ such that the diameter of $G_1 \cup H_1$ is at most 
$$
\texttt{diameter}(G_1 \cup H_1) \leq \texttt{S}(|V_1|, |H_1|) \leq 
\texttt{S}(n/x, p/x) \leq  \texttt{S}(n, p). $$

Now we claim that this implies that $\texttt{diameter}(G \cup H_1) \leq x \cdot \texttt{S}(n, p) \cdot 16 \log n$. For every pair of nodes $u, v \in V$ such that $v$ is reachable from $u$ in $G$, fix a  shortest $u \leadsto v$-path $\pi_{u, v}$ in $G$.  Then the following statement holds with high probability: for all pairs of nodes $u, v \in V$ such that $\dist_G(u, v) \geq  4x \log n$, path $\pi_{u, v}$ contains a node that was sampled into $V_2$, i.e.  $\pi_{u, v} \cap V_2 \neq \emptyset$. From now on, we will assume this property holds for our sampled set $V_2$.

Consider a pair of nodes $u, v \in V$ such that $4x \log n \leq \dist_G(u, v)$ and $v$ is reachable from $u$ in $G$, and let $\pi$ be an $u \leadsto v$-path in $G$. Let $u'$ be be the node in $\pi \cap V_1$ closest to $u$, and let $v'$ be the node in $\pi \cap V_1$ closest to $v$. By our construction of $G_1$ and the above discussion, there must be a $u' \leadsto v'$-path in $G_1$. Moreover, $\dist_{G_1 \cup H_1}(u', v') \leq \texttt{diameter}(G_1 \cup H_1) \leq \texttt{S}(n, p)$. Then since  $\dist_G(x, y) \leq 8x \log n$ for all  edges $(x, y) \in E_1$ in $G_1$, it follows that $$\dist_{G \cup H_1}(u', v') \leq \dist_{G_1 \cup H_1}(u', v') \cdot 8x \log n \leq \texttt{S}(n, p) \cdot 8x \log n.$$
Then putting it all together,
\begin{align*}
\dist_{G \cup H_1}(u, v) & \leq \dist_{G \cup H_1}(u, u') + \dist_{G \cup H_1}(u', v') + \dist_{G \cup H_1}(v', v) \\
& \leq 4x \log n + \dist_{G \cup H_1}(u', v') + 4 x \log n \\
& \leq  \dist_{G_1 \cup H_1}(u', v') \cdot 8x \log n + 8 x \log n \\
& \leq \texttt{S}(n, p) \cdot 8x \log n + 8x \log n \\
& \leq x \cdot \texttt{S}(n, p) \cdot 16 \log n.
\end{align*}

Finally, to conclude the analysis, we made two assumptions: (1) that $|V_1| \le n/x$, which occurs with probability $1/2$, and (2) that for all $u, v$ that are sufficiently far apart, we sampled a node on a $u \leadsto v$ path, which happens with high probability, i.e., $1 - 1/\text{poly}(n)$.
By an intersection bound, there is positive probability that both events happen at the same time.  So a graph $G_1$ exists as described, completing the proof.
\end{proof}